\documentclass[runningheads,orivec]{llncs}
\usepackage{varwidth}
\usepackage{hhline}
\usepackage{arydshln}

\usepackage{wrapfig}
\usepackage{float} 

\usepackage[utf8]{inputenc}
\usepackage[english]{babel}

\usepackage{graphicx}
\usepackage{transparent} 
\usepackage{makecell} 
\usepackage{booktabs}
\usepackage{array}

\usepackage{tikz} 
\usetikzlibrary{automata, positioning, arrows, calc}
\usepackage{float}
\graphicspath{{Images/}} 

\usepackage{silence} 
\usepackage{subcaption}
\usepackage{caption} 
\usepackage[normalem]{ulem} 

\usepackage{amssymb}
\usepackage{amsmath}
\usepackage[overload]{empheq} 

\usepackage{tabularx}
\usepackage{longtable} 
\usepackage{colortbl}
\usepackage{multirow}
\usepackage[ruled, vlined, linesnumbered]{algorithm2e}

\usepackage[
colorlinks=true,
linkcolor=black,
anchorcolor=black,
citecolor=black,
filecolor=black,
menucolor=black,
runcolor=black,
urlcolor=black]{hyperref} 

\usepackage{appendix}

\usepackage[inline]{enumitem}

\usepackage{lineno}

\usepackage{ifthen}
\usepackage{orcidlink}
\usepackage{pdfpages}
\usepackage{comment} 
\usepackage{tcolorbox} 
\usepackage{pgfplots}
\pgfplotsset{compat=1.18}
\usepackage{etoolbox}

\usepackage{acronym}
\usepackage{dirtree}

\usepackage[textwidth=3.5cm,textsize=footnotesize]{todonotes}

\newcommand{\todoAM}[2]{\todo[color=red!50,#1]{\textbf{AM:}#2}}

\newcommand{\Red}[1]{{\color{black}#1}}

\newcommand{\resetClock}[1]{\mathbf{#1 \boldsymbol{:=} 0}}

\newcommand{\reg}{\mathcal{R}}
\newcommand{\breg}{\boldsymbol{R}}

\newcommand{\rega}{\mathcal{R}_{\ca}}
\newcommand{\regb}{\mathcal{R}_{\cb}}
\newcommand{\regc}{\mathcal{R}_{\cc}}
\newcommand{\regu}{\mathcal{R}_{\cu}}

\newcommand{\ca}{Z}
\newcommand{\cb}{P}
\newcommand{\cc}{M}
\newcommand{\cu}{U}

\newcommand{\loc}{q}

\newcommand{\cmax}{c_m}

\newcommand{\act}{\mathit{Act}}

\newcommand{\aut}{\mathcal{A}}

\newcommand{\transta}{\mathit{TS}(\aut)}

\newcommand{\ackc}{X_{\aut}}

\newcommand{\valZ}{v[Y \leftarrow 0]}

\newcommand{\eval}{\mathit{Eval}_{\aut}}

\newcommand{\frc}{\mathit{fr}}

\newcommand{\lrf}[1]{\lfloor #1 \rfloor}

\newcommand{\res}{\mathit{res}}

\newcommand{\allclocks}[1]{\mathbf{X}_{#1}}

\newcommand{\invariant}{\mathit{Inv}}

\newcommand{\lidx}{\boldsymbol{\ell}}
\newcommand{\ridx}{\mathfrak{r}}

\newcommand{\toolname}{\textsc{Tarzan}}

\newcommand{\boolea}{\texttt{boolean}}
\newcommand{\flower}{\texttt{flower}}
\newcommand{\gates}{\texttt{gates}}
\newcommand{\ring}{\texttt{ring}}
\newcommand{\exSITH}{\texttt{exSITH}}
\newcommand{\csma}{\texttt{CSMA/CD}}
\newcommand{\pagerank}{\texttt{pagerank}}
\newcommand{\fischer}{\texttt{fischer}}
\newcommand{\medical}{\texttt{medicalWorkflow}}
\newcommand{\mpeg}{\texttt{mpeg2}}
\newcommand{\lynch}{\texttt{lynch}}
\newcommand{\train}{\texttt{trainAHV93}}
\newcommand{\bridge}{\texttt{bridge}}
\newcommand{\latch}{\texttt{latch}}
\newcommand{\maler}{\texttt{maler}}
\newcommand{\rcp}{\texttt{rcp}}
\newcommand{\soldiers}{\texttt{soldiers}}
\newcommand{\srlatch}{\texttt{srLatch}}
\newcommand{\simple}{\texttt{simple}}
\newcommand{\andor}{\texttt{andOrOriginal}}

\acrodef{ta}[TA]{Timed Automaton}
\acrodefplural{ta}[TA]{Timed Automata}
\acrodef{tarena}[TAr]{Timed Arena}
\acrodefplural{tarena}[TArs]{Timed Arenas}
\acrodef{tg}[TG]{Timed Game}
\acrodefplural{tg}[TGs]{Timed Games}
\acrodef{ltl}[LTL]{Linear Temporal Logic}
\acrodef{cltloc}[CLTLoc]{Constraint Linear Temporal Logic over clocks}
\acrodef{cltl}[CLTL]{Constraint LTL}
\acrodef{ctl}[CTL]{Computation Tree Logic}
\acrodef{tctl}[TCTL]{Timed Computation Tree Logic}
\acrodef{tol}[TOL]{Timed Obstruction Logic}
\acrodef{rts}[RTS]{Region Transition System}
\acrodefplural{rts}[RTSs]{Region Transition Systems}
\acrodef{gba}[GBA]{Generalized Büchi Automaton}
\acrodefplural{gba}[GBA]{Generalized Büchi Automata}
\acrodef{dbm}[DBM]{Difference Bound Matrix}
\acrodefplural{dbm}[DBMs]{Difference Bound Matrices}
\acrodef{bdd}[BDD]{Binary Decision Diagram}
\acrodefplural{bdd}[BDDs]{Binary Decision Diagrams}

\acrodef{bfs}[BFS]{breadth-first}
\acrodef{dfs}[DFS]{depth-first}
\acrodef{tarzan}[TARZAN]{Timed Automata Region and Zone library for real-time systems ANalysis}
\acrodef{dsl}[DSL]{Domain-Specific Language}
\acrodef{vt}[VT]{Verification Time}
\acrodef{et}[ET]{Execution Time}

\AtBeginDocument{

\newcommand*{\myeqref}[2][Equation~]{%
  \hyperref[{#2}]{#1(\ref*{#2})}%
}
\def\equationautorefname#1#2\null{%
  Equation#1(#2\null)%
}
}

\usepackage{wrapfig}

\newtoggle{EXTENDED_VERSION}
\toggletrue{EXTENDED_VERSION}

\begin{document}

\iftoggle{EXTENDED_VERSION}{
\title{\textnormal{\toolname{}}: A Region-Based Library for Forward and Backward Reachability of Timed Automata (Extended Version)\thanks{The experimental results presented in this paper are reproducible using the artifact available on Zenodo~\cite{manini_2026_tarzan}.}}}
{\title{\textnormal{\toolname{}}: A Region-Based Library for Forward and Backward Reachability of Timed Automata\thanks{The experimental results presented in this paper are reproducible using the artifact available on Zenodo~\cite{manini_2026_tarzan}.}}}

\author{
Andrea~Manini\orcidID{0009-0009-6103-9244}
\and 
Matteo~Rossi\orcidID{0000-0002-9193-9560}
\and 
Pierluigi~{San~Pietro}\orcidID{0000-0002-2437-8716}
}
\authorrunning{A. Manini \and M. Rossi\and P. {San Pietro}} 
\institute{Politecnico di Milano \\ 
\email{firstname.lastname@polimi.it}}
\titlerunning{\toolname{}: A Region-Based Library for Timed Automata}

\maketitle

\begin{abstract}
The zone abstraction, widely adopted for its notable practical efficiency, is the de facto standard in the verification of Timed Automata (TA).
Nonetheless, region-based abstractions have been shown to outperform zones in specific subclasses of TA.
To complement and support mature zone-based tools, we introduce \toolname{}, a C\texttt{++} region-based verification library for TA.
The algorithms implemented in \toolname{} use a novel region abstraction that tracks the order in which clocks become unbounded.
This additional ordering induces a finer partitioning of the state space, enabling backward algorithms to avoid the combinatorial explosion associated with enumerating all ordered partitions of unbounded clocks, when computing immediate delay predecessor regions.

We validate \toolname{} by comparing forward reachability results against the state-of-the-art tools Uppaal and TChecker.
The experiments confirm that zones excel when TA have large constants and strict guards.
In contrast, \toolname{} exhibits superior performance on closed TA and TA with punctual guards.
Finally, we demonstrate the efficacy of our backward algorithms, establishing a foundation for region-based analysis in domains like Timed Games, where backward exploration is essential.

\end{abstract}

\acresetall

\section{Introduction}
\label{sec:introduction}
Formally verifying the behavioral correctness of real-time systems (also including distributed systems) is essential 
for ensuring that critical properties hold.
Since their introduction, \acp{ta} have become one of the most widely adopted formalisms for modeling and analyzing such systems~\cite{distributed2,distributed}.

A key challenge in the verification of \ac{ta} is their infinite state space, addressed by the \emph{region}~\cite{ALURDILL} abstraction, which provides a finite representation enabling algorithmic verification. Alternatives include \emph{Binary Decision Diagrams}~\cite{Wang2004} and \emph{zones}~\cite{HENZINGER1994193}, the latter requiring \emph{normalization} for state space finiteness~\cite{Bouyer2004}. 
Zones, typically implemented using \emph{\acp{dbm}}~\cite{DBM}, are widely used in tools such as Uppaal~\cite{UppaalTutorial}, TChecker~\cite{tcheckerpaper}, Kronos~\cite{KRONOS}, and Synthia~\cite{SYNTHIA}.


\Red{Zones have shown to be highly efficient in practice and, over the last decades, have become the de facto standard for the formal verification of \acp{ta}, although they have well-known limitations~\cite{quantitativebouyer,Lehmann2023,Yovine1998}.
Yet, prior work has demonstrated that for the class of \emph{closed} \acp{ta}, \emph{i.e.}, \ac{ta} containing only non-strict guards, region-based abstractions may outperform state-of-the-art zone-based tools~\cite{J_rgensen_2012}.}

\Red{The promising performance of region-based abstractions on appropriate subclasses of \acp{ta} motivated the development of \toolname{} (Timed Automata Region and Zone library for real-time systems ANalysis)~\cite{tarzan2025}, a C\texttt{++}20 library designed to support and complement state-of-the-art zone-based tools for the formal verification of \acp{ta}.}
Currently, \toolname{} supports only region-based reachability of \ac{ta}; optimizations leveraging zones are deferred to future work.
The implementation of backward reachability algorithms enables the extension of \toolname{} to domains in which backward analysis is essential, such as Timed Games~\cite{PnueliTG}.

\Red{When compared against zone-based state-of-the-art tools including Uppaal and TChecker, \toolname{} exhibits superior performance on \ac{ta} and networks of \ac{ta} featuring \emph{punctual} guards, \emph{i.e.}, guards with equality constraints, and on closed \ac{ta}, even with large constants or many clocks.
However, the size of constants remains a significant bottleneck in other classes of \ac{ta} (\emph{e.g.}, zones excel in \ac{ta} with strict guards).
Moreover, since \toolname{} only supports reachability analysis, using forward reachability to verify safety properties for unreachable regions requires exploring the entire reachable state space.
We show in this paper that backward reachability can enhance the verification of safety properties.}

\textbf{Related works}~Abstractions closely related to regions have already been proposed to better capture the semantics of \ac{ta}.
In~\cite{TowardsUnbounded}, regions are encoded as Boolean formulae via discrete interpretations.
The \emph{slot-based} abstraction proposed in~\cite{2024parameterized} groups regions by time using integer-bounded intervals over a global clock.
A data structure, called \emph{time-dart}, is introduced in~\cite{J_rgensen_2012} to compactly encode configurations of \ac{ta} and their time successors without explicit enumeration.
In~\cite{Bouyer2008}, regions of bounded \ac{ta} are represented as tuples combining integer clock values and a partition of clocks that reflects the ordering of their fractional parts.
The region abstraction is extended in~\cite{math12244008} to \emph{durational-action} \ac{ta}.

To the best of our knowledge, no region-based library for \acp{ta} verification, nor any libraries or tools based on the aforementioned abstractions, currently exist. 
Consequently, we compare \toolname{} against the state-of-the-art tools Uppaal and TChecker.
Existing libraries that explicitly support the analysis of \ac{ta} using \ac{dbm}-based zone representations include~\cite{PyDBM,UDBM}.
Other libraries offer more general polyhedral abstractions beyond zones, such as those provided by~\cite{apron,PPL}.

\textbf{Contributions}~The main contribution of this work is two-fold:
(i) we first introduce a refined region representation that tracks the order in which clocks become unbounded.
While this representation may increase the total number of regions by inducing a finer partitioning of the state space, a major advantage is that it guarantees at most three immediate delay predecessors for any region;
(ii) building on this new representation, we develop and implement algorithms in \toolname{} for computing delay and discrete successors and predecessors of regions.
\Red{Existing tools may integrate \toolname{} to enhance their verification capabilities.}

This paper is organized as follows:
\autoref{sec:theoretical_background} introduces the necessary theoretical background.
\autoref{sec:data_structure} presents the proposed new region representation, from which a classification of regions is derived.
\autoref{sec:delay_and_discrete_predecessors} details the algorithms (and their pseudocode) to compute discrete successors and immediate delay predecessors of regions using this representation.
\autoref{sec:implementation_and_experimental_evaluation} presents an overview of \toolname{} and provides an empirical evaluation of forward and backward reachability.
\autoref{sec:conclusion} concludes and outlines future developments.
An appendix provides additional algorithms, proofs, and full experimental results.

\section{Theoretical Background}
\label{sec:theoretical_background}
The timing behavior of real-time systems is naturally modeled by the formalism of TA, which is obtained by extending classical finite-state automata with real-valued clocks. In the following, we use the term action to refer to the events that a real-time system can exhibit.
\acp{ta} are defined as follows~\cite{ALURDILL}:

\begin{definition}
\label{def:tadef}
{A \emph{\acl{ta}} $\aut$ is a tuple $\aut = (\act, Q, Q_{0}, \ackc, T, \invariant)$, where
$\act$ is a finite set of actions,
$Q$ is a finite set of locations,
$Q_{0} \subseteq Q$ is a set of initial locations,
$\ackc$ is a finite set of clocks,
$T \subseteq Q \times \act \times \Gamma(\ackc) \times 2^{\ackc} \times Q$ is a transition relation, and
$\invariant : Q \rightarrow \Gamma(\ackc)$ is an invariant-assignment function.}
\end{definition}

Let $\aut$ be a \ac{ta}.
Clocks are special variables that can only be reset or compared against non-negative integer constants. 
The value of all clocks is initialized to zero and grows with derivative 1 until a reset occurs.
When a clock is reset, its value becomes zero; after a reset, it starts increasing again. 
$\Gamma(\ackc)$ is the set of \emph{clock guards}, \emph{i.e.}, conditions
{over clocks} that must be satisfied for a transition to fire.
Clock guards are defined by the following grammar: 
$\gamma := 
\xi
\; | \;
\gamma \land \gamma$, 
where 
$\xi := 
x \sim c$ 
is a \emph{clock constraint}, $x \in \ackc$, $c \in \mathbb{N}$, and $\sim \; \in \{ \le, <, =, >, \ge \}$.
We write $\xi \in \gamma$ to indicate that the clock constraint $\xi$ occurs within the guard $\gamma$.
For instance, if $\gamma := x_1 < c \land x_2 \ge c$, then $\xi_1 := x_1 < c$ and $\xi_2 := x_2 \ge c$, where $\xi_1, \xi_2 \in \gamma$. 
The powerset $2^{\ackc}$ indicates that a transition may reset a {subset} of the clocks.
We denote by $\cmax \in \mathbb{N}$ the maximum constant appearing in $\aut$ and assume, without loss of generality, that it is the same for all clocks of $\aut$.

A pair $(q, v)$ is a \emph{configuration} of a \ac{ta} $\aut$, where $q \in Q$ and $v : \ackc \rightarrow \mathbb{R}_{\ge 0}$ is a \emph{clock valuation} function.
We denote by $\eval$ the set of all clock valuations of $\aut$.
Given a subset $Y \subseteq \ackc$ of clocks, we denote by $\valZ$ the clock valuation defined as $\valZ(x) = 0$ if $x \in Y$ and $\valZ(x) = v(x)$ otherwise.
Given a constant $\delta \in \mathbb{R}_{\ge 0}$ and a clock $x \in \ackc$, the value $(v + \delta)(x)$ is defined as $v(x) + \delta.$

{The \emph{transition system} $\transta = (S, \act \cup \mathbb{R}_{\ge 0}, I, \hookrightarrow)$ captures the semantics of a \ac{ta} $\aut$, where 
$S = Q \times \eval$,
$I = \{ (q, v) \; \vert \; q \in Q_0 \land \forall x \in \ackc : v(x) = 0 \}$, and 
$\hookrightarrow$ is the transition relation containing:
(i) \emph{discrete} transitions $((q,v), a, (q',v'))$ if there exists a transition $(q, a, \gamma, Y, q') \in T$ in $\aut$ such that $v$ satisfies both $\gamma$ and $\invariant(q)$, and $v'(x) = \valZ(x)$ satisfies $\invariant(q')$, where $a \in \act$, for all $x \in \ackc$, and
(ii) \emph{delay} transitions $((q,v), \delta, (q,v'))$ if there exists a constant $\delta \in \mathbb{R}_{\ge 0}$ such that $v'(x) = (v + \delta)(x)$ satisfies $\invariant(q)$, for all $x \in \ackc$.
A full description of the standard transition system semantics of \ac{ta} is covered in~\cite[Chapter 9]{Baier}.}

The infinite state space of \acp{ta} can be conveniently analyzed by resorting to a finite abstraction, known as the \emph{region} abstraction~\cite{ALURDILL}.
For any real number $m \in \mathbb{R}$, we denote its integral and fractional parts by $\lrf{m}$ and $\frc(m)$, respectively.

\begin{definition}[Clock equivalence relation]
\label{def:clock_equivalence}
Let $\aut\! =\! (\act, Q, Q_{0}, \ackc, T, \invariant)$ be a \ac{ta}.
Two clock valuations $v, v' \in \eval$ are called \emph{clock-equivalent}, denoted by $v \cong v'$, if, and only if, it either holds that:
\emph{(i)} $v(x) > \cmax$ and $v'(x) > \cmax$, for all $x \in \ackc$, or 
\emph{(ii)} 
for any $x, y \in \ackc$, with $v(x), v'(x) \le \cmax$ and $v(y), v'(y) \le \cmax$,
it holds that:
\emph{(ii$_a$)} $\lrf{v(x)} = \lrf{v'(x)}$ and $\frc(v(x)) = 0$ $\Leftrightarrow$ $\frc(v'(x)) = 0$, and
\emph{(ii$_b$)} $\frc(v(x)) \sim \frc(v(y))$ $\Leftrightarrow$ $\frc(v'(x)) \sim \frc(v'(y))$, where $\sim \; \in \{<, =\}$.
\end{definition}

\begin{definition}[Region]
Let $\aut$ be a \ac{ta} and $s = ( q, v )$ be a configuration of $\aut$.
The \emph{clock region} $[v]$ of $v \in \eval$ is the set of all clock valuations equivalent to $v$ under the relation of \autoref{def:clock_equivalence}.
The \emph{state region} $[s]$ of $s$ is: $[s] = ( q, [v] ) = \{ ( q, v' ) \mid v' \in [v] \}$;
$[s]$ is \emph{initial} iff $q \in Q_0$ and $\forall v' \in [v], \forall x \in \ackc : v'(x) = 0$.
\end{definition}

In the remainder of this work, we focus on state regions, referred to simply as \emph{regions}.
\Red{Let $\aut$ be a \ac{ta} and $[s] = ( q, [v] )$ a region of $\aut$.
A region $[s'] = ( q', [v'] )$ is an \emph{immediate delay successor} of $[s]$ if either:
(i) $q = q' \land [v] = [v']$ and, for all $\bar{v} \in [v]$ and $x \in \ackc$, $\bar{v}(x) > \cmax$, or
(ii) $q = q' \land [v] \neq [v']$ and 
for all $\bar{v} \in [v]$ 
there exists $\delta \in \mathbb{R}_{> 0}$
such that $\bar{v} + \delta \in [v'] 
\land 
\forall \; 0 \le \delta' \le \delta : \bar{v} + \delta' \in ([v] \cup [v'])$.
A region $[s']$ is a \emph{delay successor} of $[s]$ if there exists a finite sequence of regions $[s_0], [s_1], \dots, [s_n]$ such that $[s_0] = [s]$, $[s_n] = [s']$, and each $[s_{i + 1}]$ is an immediate delay successor of $[s_i]$, for $0 \le i < n$.
A region $[s']$ is a \emph{discrete successor} of $[s]$ if $[v'] = [\valZ]$, where $Y \subseteq \ackc$ is the subset of clocks reset over the discrete transition from $[s]$ to $[s']$.
Symmetric definitions can be given for \emph{immediate delay predecessors}, \emph{delay predecessors}, and \emph{discrete predecessors}.}

\section{Representing Regions with Ordered Clock Partitions}
\label{sec:data_structure}
We introduce a new region representation that captures the order in which clocks become unbounded (exceeding $\cmax$).
For bounded clocks (with value up to $\cmax$), only the ordering of their fractional part is preserved.
\Red{Our representation builds on~\cite{Bouyer2008}, which focuses on \acp{ta} with bounded clocks.
By also handling unbounded clocks, we significantly extend the applicability of region-based state space exploration to a broader class of \acp{ta}.}
{While forward exploration remains unaffected, our representation influences backward exploration (\autoref{thm:bounded_d_pred} and \autoref{thr:discrete_predecessor_complexity}).}

Let $\aut = (\act, Q, Q_{0}, \ackc, T, \invariant)$ be a \ac{ta}, $q \in Q$.
Furthermore, let $\lidx, \ridx \in \mathbb{N}$ be two integers such that $\lidx + \ridx \le \lvert \ackc \rvert$.
We denote by $X_i \subseteq \ackc$ a subset of clocks of $\aut$, for $-\lidx \le i \le \ridx$.
\Red{We require that these subsets induce a partition of $\ackc$, \emph{i.e.}, $\ackc = \bigcup_{i} X_i$ and, for all indices $-\lidx \le i,j \le \ridx$ such that $i \neq j$, $X_i \cap X_j = \varnothing$.}
Finally, we introduce a function $h : \ackc \rightarrow \{0, \dots, \cmax \}$ that assigns an integer value to each clock in $\ackc$.
A region of $\aut$ can be represented as follows:
\begin{equation}
\label{eq:region}
    \reg := \{ \loc, h,  X_{-\lidx}, X_{-(\lidx-1)}, \dots, X_{-1}, X_{0}, X_1, \dots, X_{(\ridx-1)}, X_{\ridx} \}.
\end{equation}

Set $X_0$ is always present even if empty; the other sets are omitted when empty.
The sets $X_{-\lidx}, \dots, X_{-1}$ represent unbounded clocks, capturing the ordering in which these clocks become unbounded
(\emph{e.g.}, all clocks in $X_{-1}$ become unbounded simultaneously and prior to all clocks in $X_{-2}$).  
Conversely, the sets $X_1, \dots, X_{\ridx}$ represent bounded clocks, capturing the ordering of their fractional parts (\emph{e.g.}, the fractional parts of all clocks in $X_{1}$ are equal and less than the fractional parts of all clocks in $X_{2}$).
$X_0$ contains bounded clocks that have no fractional part (we also say these clocks \emph{are in the unit}).
The integers $\lidx$ and $\ridx$ denote the number of sets containing unbounded and bounded clocks, respectively ($\ridx$ does not consider $X_0$).  
Their values are automatically updated whenever a set is inserted into or removed from $\reg$.
We use a subscript to disambiguate between different regions: for instance, $q_{\reg}$ refers to the location $q$ of $\reg$, $X_{\reg,0}$ to the set $X_0$ of $\reg$, and $\lidx_{\reg}$ indicates the number of sets containing unbounded clocks in $\reg$.

\Red{For conciseness, we also refer to the elements of a region $\reg$ as its \emph{structure}.
Two regions $\reg,\reg'$ are \emph{structurally equivalent} if the following holds: 
$q_{\reg} = q_{\reg'} 
\land 
\lidx_{\reg} = \lidx_{\reg'}
\land
\ridx_{\reg} = \ridx_{\reg'}$
and, for all $-\lidx_{\reg} \le i \le \ridx_{\reg}$, $X_{\reg, i} = X_{\reg', i}$.}

\begin{example}
\label{ex:exampleregion}
    Let $\aut$ be a \ac{ta} with clocks $\ackc = \{ x, y, z, w \}$ and $\cmax = 5$.
    Consider the region defined as:
    $\{ (x,y,z,w) \; \vert \; 2 < x < 3, 2 < y < 3, z > 5, w > 5, x = y \}$,
    \emph{i.e.}, a region (assumed to be in $q \in Q$) in which $x$ and $y$ are bounded and equal (they have the same fractional part),
    while clocks $z$ and $w$ are unbounded (we assume that $z$ became unbounded before $w$).
    The region is represented using \autoref{eq:region} as:
    $
    \reg = \{ q, h(x) = 2, h(y) = 2, h(z) = 5, h(w) = 5, X_{-2}, X_{-1}, X_0, X_1 \},
    $
    where $X_{-2} = \{ w \}, X_{-1} = \{ z \}, X_0 = \varnothing, X_1 = \{ x, y \}, \lidx = 2$, and $\ridx = 1$.
    
\end{example}

\begin{definition}[Region class]
\label{def:region_classes}
    If the order in which clocks become unbounded is tracked, regions can be classified in four distinct \emph{classes}---\ca{} (Zero), \cb{} (Positive), \cc{} (Mixed), and \cu{} (Unbounded)---defined as:
    \emph{(i)} Class \ca:~$
    \lidx \ge 0 \land \ridx = 0 \land X_0 \neq \varnothing$;
    \emph{(ii)} Class \cb:~$\lidx \ge 0 \land \ridx > 0 \land X_0 = \varnothing$; 
    \emph{(iii)} Class \cc:~$\lidx \ge 0 \land \ridx > 0 \land X_0 \neq \varnothing$;
    and \emph{(iv)} Class \cu:~$\lidx > 0 \land \ridx = 0 \land X_0 = \varnothing$.
\end{definition}

Class $\ca{}$ denotes regions in which some clocks may be unbounded, while all bounded clocks have no fractional part and at least one clock must be bounded.
Class $\cb{}$ denotes regions where some clocks may be unbounded, and all bounded clocks have a fractional part greater than zero (again, at least one clock must be bounded).
Class $\cc{}$ denotes regions where some clocks may be unbounded and the bounded clocks are partitioned into those with no fractional part and those with a fractional part greater than zero (at least two clocks must be bounded, one with no fractional part and one with a fractional part greater than zero).
Finally, class $\cu{}$ denotes regions in which all clocks are unbounded.
Notice that the region presented in \autoref{ex:exampleregion} belongs to class $\cb$.
This classification into different classes lies at the core of the algorithms that compute the immediate delay successor and immediate delay predecessors of a given region.

\textbf{Transitioning between classes}~Let $\rega$, $\regb$, $\regc$, and $\regu$ be regions of classes $\ca$, $\cb$, $\cc$, and $\cu$, respectively. 
Transitions from a region to its immediate delay successor
induce the following deterministic transitions between classes:
\begin{itemize}[topsep=0pt, partopsep=0pt, parsep=0pt, itemsep=0pt]
    \item \emph{From $\rega$ to $\regb$}: at least one clock $x \in X_{\rega,0}$ must increase such that its fractional part becomes greater than zero, while remaining bounded. 
    Other clocks in $X_{\rega,0}$ can either become unbounded or remain bounded and grow at the same rate of $x$, \emph{i.e.}, all bounded clocks have the same fractional part.
    
    \item \emph{From $\rega$ to $\regu$}: all clocks in $X_{\rega,0}$ become unbounded simultaneously.
    
    \item \emph{From $\regb$ to $\rega$}: it must hold that $\ridx_{\regb} = 1$.
    All clocks in $X_{\regb, 1}$
    must reach the next unit (they will be inserted into $X_{\rega,0}$) and their integer value is incremented by 1, as they have a fractional part greater than zero.

    \item \emph{From $\regb$ to $\regc$}: it must hold that $\ridx_{\regb} > 1$.
    All clocks in $X_{\regb, \ridx_{\regb}}$ must reach the next unit (they will be inserted into $X_{\regc,0}$) and their integer value is incremented by 1, \Red{as they have the biggest fractional part.}

    \item \emph{From $\regc$ to $\regb$}: clocks in $X_{\regc,0}$ can become unbounded or their fractional part can become greater than zero while remaining bounded.
\end{itemize}
The missing transitions between classes are explained as follows.
By \autoref{def:region_classes}, transitioning to $\regu$ is only possible from $\rega$, since $\ridx_{\regu} = 0 \land X_{\regu, 0} = \varnothing$ must hold.
It is not possible to transition from $\rega$ to $\regc$, since all clocks will exit from $X_{\rega, 0}$.
An immediate delay successor of $\regc$ has at least one clock with a nonzero fractional part, making it impossible to transition from $\regc$ to $\rega$.

Transitions between classes contain no self-loops; the immediate delay successor of any region is always a distinct region.
$\regu$ does not have delay successors, as clocks can increase indefinitely once unbounded.
The structure of regions influences transitioning between classes;
\emph{e.g.}, from $\regc$ with $\ridx_{\regc} = 3$, several cycles between $\regb$ and $\regc$ are required before reaching $\rega$ via delay transitions; this becomes possible once $\ridx_{\regc}$ decreases to $1$ as some clocks become unbounded.

\section{Discrete Successors and Immediate Delay Predecessors}
\label{sec:delay_and_discrete_predecessors}
We present algorithms computing discrete successors and immediate delay predecessors regions.
For brevity, immediate delay successors and discrete predecessors algorithms are deferred to the Appendix in \autoref{app:foreword_delay_successors_algorithm} and \autoref{app:discrete_predecessors_algorithm}.

\textbf{Discrete successors}~\autoref{alg:finddiscretesucc}, whose time complexity is $O(\lvert T \rvert \cdot \lvert \ackc \rvert)$, computes the discrete successors of a given region $\reg$.
In particular, a set storing the resulting discrete successors is initialized in Line~\ref{line:fdsuc:initres} and returned in Line~\ref{line:fdsuc:returnres}.  
The loop in Line~\ref{line:fdsuc:foreacht} iterates over all transitions $t$ of the original \ac{ta} that exit from the current location $q_{\reg}$ of $\reg$, identifying all possible discrete successors of $\reg$.  
If the guard $\gamma$ of $t$ is satisfied by combining the integer values of the clocks of $\reg$, as returned by \( h_{\reg} \), with their fractional parts (Line~\ref{line:fdsuc:isenabled}), a copy $\breg$ of $\reg$ is created to represent the successor region (Line~\ref{line:fdsuc:copyr}).  
The location of $\breg$ is updated to the target location $q$ of $t$ (Line~\ref{line:fdsuc:setloc}).  
The loop in Line~\ref{line:fdsuc:idxloop} iterates over all integers $-\lidx_{\breg} \le i \le \ridx_{\breg}$, enabling the inner loop in Line~\ref{line:fdsuc:setloop} to examine each clock set of $\breg$.  
If a clock $x$ is reset in $t$,
its integer value is set to zero (Line~\ref{line:fdsuc:setxtozero}).  
Moreover, if $x$ is either unbounded or has a nonzero fractional part (Line~\ref{line:fdsuc:hasfracpart}), it is inserted into $X_{\breg, 0}$ and removed from its previous clock set (Line~\ref{line:fdsuc:insertremovex}).  
Finally, the resulting region $\breg$ is added to the result set in Line~\ref{line:fdsuc:addtores}.

\begin{algorithm}[thb]
    \caption{\texttt{find-discrete-successors}$(\reg)$}
    \label{alg:finddiscretesucc}
    
    \KwData{$\reg$: a region.}
    \KwResult{The discrete successors of $\reg$.}
    
    \Begin() {
        $\mathit{res} \gets \varnothing$\; \label{line:fdsuc:initres}
        \ForEach{$t = (q_{\reg}, a, \gamma, Y, q) \in T$}{ \label{line:fdsuc:foreacht}
            \If{$\gamma$ is satisfied by $\reg$}{ \label{line:fdsuc:isenabled}
                $\breg \gets$ copy of $\reg$\; \label{line:fdsuc:copyr}
                $q_{\breg} \gets q$\; \label{line:fdsuc:setloc}
                \For{$i = -\lidx_{\breg}$ \KwTo $\ridx_{\breg}$}{ \label{line:fdsuc:idxloop}
                    \ForEach{$x \in X_{\breg,i}$ such that $x \in Y$}{ \label{line:fdsuc:setloop}
                            $h_{\breg}(x) \gets 0$\; \label{line:fdsuc:setxtozero}
                            \If{$i \neq 0$}{ \label{line:fdsuc:hasfracpart}
                                $X_{\breg,0} \gets X_{\breg,0} \cup \{x\}$;
                                $X_{\breg,i} \gets X_{\breg,i} \setminus \{x\}$\; \label{line:fdsuc:insertremovex}
                        }
                    }
                }
                $\mathit{res} \gets \mathit{res} \cup \{ \breg \}$\; \label{line:fdsuc:addtores}
            }
        }
        \Return $\mathit{res}$\; \label{line:fdsuc:returnres}
    }
\end{algorithm}

\begin{algorithm}[t]
    \caption{\texttt{find-immediate-delay-predecessors}$(\reg)$}
    \label{alg:finddelaypred}
    
    \KwData{$\reg$: a region where no clock $x$ is such that: $h_{\reg}(x) = 0 \land x \in X_{\reg, 0}$.}
    \KwResult{The immediate delay predecessors of $\reg$.}
    \Begin() {
        $\res \gets \varnothing$\; \label{line:dp:declarereg}
        $\breg \gets$ copy of $\reg$\; \label{line:dp:copyr}
        \uIf{$X_{\reg,0} = \varnothing \land \ridx_{\reg} = 0$}{ \label{line:dp:ucheck}
            insert clocks of $X_{\breg, -\lidx_{\breg}}$ in $X_{\breg, 0}$ and remove $X_{\breg, -\lidx_{\breg}}$ from $\breg$\;
            \label{line:dp:handleuclass}
        }
        \uElseIf{$X_{\reg,0} \neq \varnothing$}{ \label{line:dp:accheck}
            $X_{\mathit{tmp}} \gets \varnothing$\; \label{line:dp:initxtmp}
            \ForEach{$x \in X_{\breg,0}$}{ \label{line:dp:foreachloopac}
                $h_{\breg}(x) \gets h_{\breg}(x) - 1$\; \label{line:dp:decreaseh}
                $X_{\mathit{tmp}} \gets X_{\mathit{tmp}} \cup \{ x \}$; 
                $X_{\breg,0} \gets X_{\breg,0} \setminus \{ x \}$\; \label{line:dp:insertandremovex}
            }
            insert $X_{\mathit{tmp}}$ to the right of $X_{\breg, \ridx_{\breg}}$\; \label{line:dp:insertxtmpright}
        }
        \Else{
        \label{line:dp:bcheck}
            insert clocks of $X_{\breg, 1}$ in $X_{\breg, 0}$ and remove $X_{\breg, 1}$ from $\breg$\; \label{line:dp:insertx1inx0}

            \If{$\lidx_{\reg} \ge 1$}{ \label{line:dp:checkunbound}
                $\breg_1 \gets$ copy of $\reg$\; \label{line:dp:copyinr1}
                insert clocks of $X_{\breg_1, -\lidx_{\breg_1}}$ in $X_{\breg_1, 0}$ and remove $X_{\breg_1, -\lidx_{\breg_1}}$ from $\breg_1$\; \label{line:dp:insertinr1}
                $\breg_2 \gets$ copy of $\breg_1$\; \label{line:dp:copyinr2}
                insert clocks of $X_{\breg_2, 1}$ into $X_{\breg_2, 0}$ and remove $X_{\breg_2, 1}$ from $\breg_2$\; \label{line:dp:insertinr2}
                $\res \gets \res \cup \{ \breg_1, \breg_2 \}$\; \label{line:dp:collectr1r2}
             }
        }
        \Return ($\res \cup \{ \breg \})$\; \label{line:dp:returnres}
    }
\end{algorithm}

\textbf{Immediate delay predecessors}~\Red{The following theorem (whose proof is shown in the Appendix, \autoref{thm:proof_of_bounded_d_pred}) establishes that, when tracking the order in which clocks become unbounded, the number of immediate delay predecessors of a given region does not exhibit exponential blow-up.}

\begin{theorem}
\label{thm:bounded_d_pred}
   If the order in which clocks become unbounded is tracked, the number of immediate delay predecessors of a given region is at most three.
\end{theorem}

\Red{Indeed, if the order in which clocks become unbounded is not tracked, computing the immediate delay predecessors of regions with unbounded clocks would require enumerating all ordered partitions of the unbounded clocks, as every order in which the clocks become unbounded would then be admissible (regions without unbounded clocks would have a unique immediate delay predecessor).}

The immediate delay predecessors of a given region $\reg$ are computed by \autoref{alg:finddelaypred}, whose time complexity is $O(\lvert \ackc \rvert)$ (we refer to \emph{predecessors} in its description).
In particular, Line~\ref{line:dp:declarereg} initializes a set, returned in Line~\ref{line:dp:returnres}, storing the computed predecessors.
The algorithm proceeds by analyzing the class of $\reg$ to compute its predecessors by modifying a copy $\breg$ of $\reg$, created in Line~\ref{line:dp:copyr}.

If $\reg$ is of class $\cu$ (Line~\ref{line:dp:ucheck}), the predecessor is obtained by inserting set $X_{\breg, -\lidx_{\breg}}$, which contains the clocks that became unbounded while exiting from the unit, into $X_{\breg, 0}$.
Then, $X_{\breg, -\lidx_{\breg}}$ is removed from $\breg$ (Line~\ref{line:dp:handleuclass}).

For regions of class $\ca$ or class $\cc$ (Line~\ref{line:dp:accheck}),
an auxiliary set $X_{\mathit{tmp}}$ is initialized (Line~\ref{line:dp:initxtmp}).
The integer part of each clock $x \in X_{\breg, 0}$ (Line~\ref{line:dp:foreachloopac}) is decremented by 1 (Line~\ref{line:dp:decreaseh}), and $x$ is moved from $X_{\breg, 0}$ to $X_{\mathit{tmp}}$ (Line~\ref{line:dp:insertandremovex}). Since these clocks must have
the highest fractional part, $X_{\mathit{tmp}}$ is inserted to the right of $X_{\breg, \ridx_{\breg}}$ (Line~\ref{line:dp:insertxtmpright}).

If $\reg$ is of class $\cb$ (Line~\ref{line:dp:bcheck}), the algorithm considers both the cases where the predecessor is of class $\ca$ or class $\cc$. 
First, clocks in $X_{\breg, 1}$ (those that left the unit) are moved into $X_{\breg, 0}$ (Line~\ref{line:dp:insertx1inx0}). Then, Line~\ref{line:dp:checkunbound} checks whether unbounded clocks must also be handled.
If so, two additional regions are created.
The first, $\breg_1$, is created on Line~\ref{line:dp:copyinr1} as a copy of $\reg$ and updated on Line~\ref{line:dp:insertinr1} to represent the case in which all clocks with zero fractional part in the predecessor became unbounded.
After $\breg_1$ is updated, it is copied in $\breg_2$ (Line~\ref{line:dp:copyinr2}), which is modified on Line~\ref{line:dp:insertinr2} to capture the case in which only a subset of those clocks became unbounded.
Both $\breg_1$ and $\breg_2$ are collected on Line~\ref{line:dp:collectr1r2} in the result set.


\renewcommand{\arraystretch}{1}
\definecolor{lightgray}{gray}{0.95}
\begin{figure}[tb]
    \centering

    \begin{subfigure}[H]{0.54\textwidth}
        \centering
        \includegraphics[width=\linewidth]{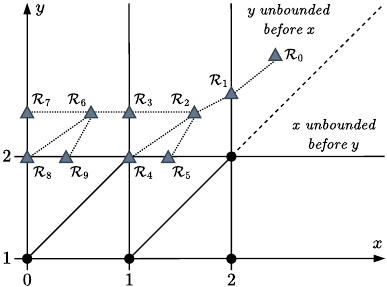}
        \caption{Partial computation of delay predecessors of $\reg_0$.}
        \label{fig:partcompfig}
    \end{subfigure}%
    \hfill
    %
    \begin{subfigure}[H]{0.44\textwidth}
        \centering
       \rowcolors{2}{lightgray}{white}  
        \begin{tabular}{
        >{\centering\arraybackslash}p{1.3cm}  
        >{\centering\arraybackslash}p{1.3cm}  
        p{2cm}                              
        }
        \toprule
        \textbf{Region} & \textbf{Class} & \textbf{Unbounded} \\
        \midrule
        $\reg_0$ & $\cu$ & yes, $\{x, y\}$ \\
        $\reg_1$ & $\ca$ & yes, $\{y\}$ \\
        $\reg_2$ & $\cb$ & yes, $\{y\}$ \\
        $\reg_3$ & $\ca$ & yes, $\{y\}$ \\
        $\reg_4$ & $\ca$ & no, $\varnothing$ \\
        $\reg_5$ & $\cc$ & no, $\varnothing$ \\
        $\reg_6$ & $\cb$ & yes, $\{y\}$ \\
        $\reg_7$ & $\ca$ & yes, $\{y\}$ \\
        $\reg_8$ & $\ca$ & no, $\varnothing$ \\
        $\reg_9$ & $\cc$ & no, $\varnothing$ \\
        \bottomrule
        \end{tabular}
        \caption{Table listing regions $\reg_i$, $0 \le i \le 9$, their class, and their unbounded clocks (if any).}
        \label{fig:partcomptab}
    \end{subfigure}

    \caption{Example of delay predecessors computation from $\reg_0$ (left), stopped upon reaching bounded regions, and a table with additional information about the regions (right).
    Triangles symbolically represent entire regions rather than specific points within them.}
    \label{fig:partcomb}
\end{figure}

\begin{example}
\label{ex:delaypred}
Let $\aut$ be a \ac{ta} having clocks $\ackc = \{ x, y \}$ and maximum constant $\cmax = 2$.  
Consider the region $\reg_0$ 
as depicted in \autoref{fig:partcompfig}.
Here, triangles are used to symbolically represent entire regions rather than individual points (\emph{e.g.}, the triangle
corresponding to $\reg_2$ represents the region 
$\{ (x, y) \; | \; 1 < x < 2 \land y > 2 \}$).
For this reason, they are not drawn on the diagonals, as would be implied by \autoref{def:clock_equivalence}.
\Red{For visual clarity, the figure does not represent the portion of the state space where $y < 1$.}
Since $\reg_0$ lies above the diagonal originating from $(2,2)$, it corresponds to the case in which $y$ became unbounded before $x$.  
$\reg_0$ can be represented as:
$
\mathcal{R}_0 = \{ q,\, h(x) = 2,\, h(y) = 2,\, X_{-2} = \{ x \},\, X_{-1} = \{ y \},\, X_0 = \varnothing \}.
$

We compute the delay predecessors of $\reg_0$ using \autoref{alg:finddelaypred}.  
For conciseness, we stop the computation when reaching a bounded region or because no further delay predecessors exist.  
The algorithm yields a single immediate delay predecessor of class $\ca$ (\autoref{fig:partcomptab}):
$
\mathcal{R}_1 = \{ q,\, h(x) = 2,\, h(y) = 2,\, X_{-1} = \{ y \},\, X_0 = \{ x \} \}
$.
Applying the same algorithm to $\reg_1$ yields a single immediate delay predecessor:
$
\mathcal{R}_2 = \{ q,\, h(x) = 1,\, h(y) = 2,\, X_{-1} = \{ y \},\, X_0 = \varnothing,\, X_1 = \{ x \} \}.
$
Then, from $\reg_2$, three immediate delay predecessors can be computed. Two of them ($\reg_4$ and $\reg_5$) are not further analyzed in this example, as they are already bounded.
Consider the region:
$
\mathcal{R}_3 = \{ q,\, h(x) = 1,\, h(y) = 2,\, X_{-1} = \{ y \},\, X_0 = \{ x \} \}.
$  
This region is structurally equivalent to $\reg_1$, except that the integer value of $x$ has been decremented by 1.
The computation from $\reg_1$ applies similarly from $\reg_3$.
\end{example}

The computation of delay predecessors can be optimized.
Assume that, for a region $\reg$, \autoref{alg:finddelaypred} can be applied 
$n$ times. 
Furthermore, assume that all clocks in $\reg$ are bounded and that $\ridx_{\reg} > 0$ holds.
%
%
We define the \emph{period} of $\reg$ as
$\theta = 2 \cdot (\ridx_{\reg} + 1)$ if $X_{\reg, 0} \neq \varnothing$; $\theta = 2 \cdot \ridx_{\reg}$ otherwise.
Then, decrement the integer value of every clock of $\reg$ by $\lfloor n / \theta \rfloor$ (ensuring their integer value does not go below zero), obtaining a new region $\reg'$.
After this step, if none of the clocks of $\reg'$ are equal to zero \Red{(with no fractional part)}, it suffices to apply \autoref{alg:finddelaypred} only $n \bmod \theta$ times to $\reg'$ to obtain the final result, possibly stopping earlier if one or more clocks become equal to zero during this process.
Similar optimizations can also be applied when some clocks are unbounded.  


\textbf{Discrete predecessors (intuition)}~When computing a discrete successor, a subset of clocks may be reset, and therefore appear with value zero in the successor region. 
Discrete predecessors computation aims to reverse this operation. 
However, starting from a region in which the value of some clocks is zero (with no fractional part), a challenge arises, since the original values of these clocks prior to being reset are unknown. 
As a result, all possible combinations of values for the reset clocks must be considered (including cases where they may have been unbounded) to ensure all discrete predecessors are computed.
This process can be simplified using the guard of the transition over which discrete predecessors are computed, as it may constrain the possible values of clocks prior to reset.

The following theorem (whose proof is shown in the Appendix, \autoref{thm:proof_of_discrete_predecessor_complexity}) establishes that, when tracking the order in which clocks become unbounded, the number of discrete predecessors of a given region suffers from an inherent exponential blow-up that cannot be avoided in general.

\begin{theorem}
\label{thr:discrete_predecessor_complexity}
    Let $\aut$ be a \ac{ta} such that $n = \lvert \ackc \rvert$ and let $\reg$ be a region of $\aut$. If the order in which clocks become unbounded must be tracked, then the number of discrete predecessors of $\reg$ over a single transition is $O\left( (\frac{\cmax + 1}{\ln 2}) ^{n}\cdot (n+1)!\right)$.
\end{theorem}


\section{Implementation and Experimental Evaluation}
\label{sec:implementation_and_experimental_evaluation}
This section introduces \toolname{} along with a detailed empirical evaluation of its forward and backward reachability algorithms.
To express reachability queries, we use the standard Computation Tree Logic formula $\exists \Diamond ( \varphi )$, meaning that there exists a path in the state space in which $\varphi$ eventually holds true.

\subsection{The \textnormal{\toolname{}} library}
\toolname{} is a C\texttt{++}20 library for performing forward and backward reachability analysis of \ac{ta} and networks of \ac{ta}, using either \ac{bfs} or \ac{dfs} state space exploration.
\toolname{} supports shared integer variables (specified by using general arithmetic expressions), urgent locations, and invariants. 
The semantics of both invariants and networks in \toolname{} follows those of Uppaal~\cite{Bengtsson2004}.
Support for networks of \ac{ta} and integer variables in backward reachability is currently under development.
Successors and predecessors algorithms were revised to allow each clock to have its own maximum constant.

\toolname{} accepts .txt input files, parsed using the Boost Spirit X3 library, containing a \ac{ta} representation written in a domain-specific language called \emph{Liana}.
Each file represents a single \ac{ta}; a network of \ac{ta} is defined using multiple files.

\textbf{Networks of \ac{ta}}~A \emph{network} of size $n$ is the parallel composition of $n$ individual \ac{ta}~\cite[Chapter 9]{Baier}.
Given $n$ regions $\reg_1, \dots, \reg_n$, each corresponding to a \ac{ta} $\aut_1, \dots, \aut_n$ in the network, a \emph{network region} $\overline{\reg}$ is defined as $\overline{\reg} := \{ \reg_1, \dots, \reg_n \}$. 
To compute both immediate delay and discrete successors of $\overline{\reg}$, \toolname{} 
{relies on specialized data structures that generalize the applicability of the algorithms originally developed for individual regions to network regions.
These data structures ensure the correctness of on-the-fly immediate delay and discrete successors computation across the Cartesian product of $\aut_1, \dots, \aut_n$.
In particular, a C\texttt{++} double-ended queue maintains the relative ordering of the fractional parts of all clocks within $\overline{\reg}$, partitioning these clocks into ordered equivalence classes.
This ordering captures the evolution over time of all clocks within $\overline{\reg}$, including which clocks will exit the unit (possibly becoming unbounded), reach the next unit, or be reset.
Furthermore, a C\texttt{++} set identifies the regions in $\overline{\reg}$ that are of class $\ca$ or $\cc$.
Since the generalized algorithms for network regions operate by applying the successors algorithms to a subset of individual regions in $\overline{\reg}$, this classification is essential for computing immediate delay successors; indeed, regions of class $\ca$ or  $\cc$ contain clocks that will exit the unit; hence, their immediate delay successors must be computed prior to those of regions of class $\cb$.}

\textbf{Implementation considerations}~Since the focus of this paper is the implementation of algorithms using the new representation of \autoref{eq:region}, most optimizations are deferred to future work.
Currently, \toolname{} does not implement any optimizations for state space exploration, except for symmetry reduction~\cite{symmspin}.
Memory compression optimizations are not implemented either.
The algorithms also suffer when computing delay successors (predecessors), as these computations require generating a sequence of immediate delay successors (predecessors).
This limitation could be mitigated by implementing the optimization described in \autoref{sec:delay_and_discrete_predecessors} for delay predecessors (which can be adapted for delay successors as well).
\Red{Furthermore, since \toolname{} only performs reachability analysis, proving that a region is unreachable going forward requires exploring the entire reachable state space.}
\Red{Due to the aforementioned limitations, a \ac{bfs} exploration must visit all delay successors (predecessors) before visiting discrete successors (predecessors), which is typically unnecessary and increases the verification time.
For this reason, unless stated otherwise, all experiments follow a \ac{dfs} strategy.}

\subsection{Forward reachability benchmarks}
\setlength{\intextsep}{0pt}
\begin{wrapfigure}[12]{r}{0.35\textwidth}
\centering
\captionsetup{skip=5pt}
\begin{tikzpicture}[
    ->,
    >=stealth',
    shorten >=2pt, 
    auto,
    transform shape,
    align=center,
    scale=0.865,
    state/.style={thick, circle, draw, minimum size=1cm}
] 
\useasboundingbox (-2, -3.1) rectangle (2.9, 1.71);

  \node[state, accepting] (q) {$q_0$};
  \node[state, below of=q, yshift=-1.5cm] (goal) {$\mathit{Goal}$};
  \node[above=1cm of q, xshift=23pt, yshift=-37pt] (dots) {\Large $\mathbf{\cdots}$};
  \draw [thick] 
  (q) edge[loop, in=180, out=240, looseness=8]  node[below, xshift=0pt, yshift=-6pt]  {$x_1 = 1$ \\ $\resetClock{x_1}$} (q)
  
  (q) edge[loop, in=90, out=150, looseness=8] node[right, xshift=21pt, yshift=0pt] {$x_2 = 2$ \\ $\resetClock{x_2}$} (q)
  
  (q) edge[loop, in=300, out=0, looseness=8] node[right, xshift=1pt, yshift=5pt] {$x_n = n$ \\ $\resetClock{x_n}$} (q)
  
  (q) edge node[right, yshift=-13pt] {$x_1 = \cdots = x_n = 0$ \\ $\land \; y \geq 1$} (goal);

\end{tikzpicture}
\caption{\flower{}. $q_0$ is initial.}
\label{fig:flower_main_text}
\end{wrapfigure}
We compared \toolname{} against Uppaal 5.0 and TChecker (commit d711ace) in forward reachability\footnote{To ensure a fair comparison, symmetry reduction was disabled in both \toolname{} and Uppaal; notably, this optimization cannot be manually controlled in TChecker.}.
The benchmarks \Red{(adjusted to ensure compatibility across all tools)} come from publicly available repositories, published literature, or manually crafted when required.
If possible, the reachability queries match those of the original benchmarks; 
otherwise, they have been adapted to fit the capabilities of the tools.
Uppaal was run with standard DBM without memory optimizations, TChecker was configured with the \texttt{covreach} flag to match the algorithm used in Uppaal.
%
Experiments were conducted on a MacBook Pro equipped with an M3 chip and 24 GB of RAM.

In the following, tables report \ac{vt}, \ac{et}, memory (Mem, representing the Maximum Resident Set Size), and the number of stored regions and states.
\Red{\ac{et} measures the time elapsed between issuing a verification command and receiving the result from a tool.}
Timings are in seconds, memory in MB.
The size of a network is denoted by $K$ (for \flower{}, $K$ indicates the number of clocks).
All values are averages over five runs, with a timeout of 600 s set on \ac{et}.
Timeouts are reported as TO, $\varepsilon$ denotes times below 0.001 s, and out-of-memory events are denoted by OOM.
Bold cells indicate the fastest \ac{vt} (on ties, the fastest \ac{et}).
KO indicates that an error occurred during verification.
The complete benchmark suite is in the Appendix, \autoref{app:complete_experimental_evaluation}.

\textbf{\ac{ta} with punctual guards}~Region-based abstractions can outperform zone-based approaches when \acp{ta} are \emph{closed}, \emph{i.e.}, contain only non-strict guards.
This is particularly evident when guards are \emph{punctual}, \emph{i.e.}, restricted to equality constraints.
The \flower{} \ac{ta}~\cite{J_rgensen_2012} depicted in \autoref{fig:flower_main_text}, which is also used in our evaluation of backward reachability, satisfies these requirements.
It consists of two locations: $q_0$ (which is initial) and $\mathit{Goal}$, and it features $n + 1$ clocks, denoted by $x_1, \dots, x_n, y$.
The guards over clocks $x_i$ are punctual; the guard on $y$ is non-strict. 
Since $y$ is never reset, the behavior of \flower{} is equivalent to a model featuring only punctual guards.
To reach $\mathit{Goal}$, all clocks $x_i$ must simultaneously be exactly zero, while $y$ must satisfy $y \ge 1$.
Specifically, when transitioning to $\mathit{Goal}$, $y$ must be equal to the least common multiple of $1, 2, \dots, n$.
Similar behavior is shown in the additional examples reported in \autoref{tab:punctual}: these are networks of \ac{ta} with only punctual guards.
In all these experiments, punctual guards force zones to be repeatedly and finely split, as reflected by the large number of states shown in \autoref{tab:punctual}.
The results indicate that Uppaal and TChecker experience a substantial decrease in performance when $K \ge 8$.

\definecolor{lightgray}{gray}{0.95}
\definecolor{lightgreen}{RGB}{210,255,190}
\setlength{\aboverulesep}{1pt}
\setlength{\belowrulesep}{1pt}
{\renewcommand{\arraystretch}{0.95}
{\setlength{\tabcolsep}{3.25pt}
\begin{table}[tb]
\centering
\caption{When all transitions have punctual guards, regions may outperform zones.
Here, \flower{} is a single \ac{ta}, while the others are networks of \ac{ta}.}
\label{tab:punctual}
\resizebox{\textwidth}{!}{
\begin{tabular}{%
    c|rrrr|rrrr|rrrr
}
\toprule
\multirow{2}{*}{\textbf{K}}
  & \multicolumn{4}{c|}{\textbf{TARZAN}}
  & \multicolumn{4}{c|}{\textbf{TChecker}}
  & \multicolumn{4}{c}{\textbf{UPPAAL}} \\
  & \multicolumn{1}{c}{\textbf{VT}}
  & \multicolumn{1}{c}{\textbf{ET}}
  & \multicolumn{1}{c}{\textbf{Mem}}
  & \multicolumn{1}{c|}{\textbf{Regions}}
  & \multicolumn{1}{c}{\textbf{VT}}
  & \multicolumn{1}{c}{\textbf{ET}}
  & \multicolumn{1}{c}{\textbf{Mem}}
  & \multicolumn{1}{c|}{\textbf{States}}
  & \multicolumn{1}{c}{\textbf{VT}}
  & \multicolumn{1}{c}{\textbf{ET}}
  & \multicolumn{1}{c}{\textbf{Mem}}
  & \multicolumn{1}{c}{\textbf{States}} \\
\midrule
\multicolumn{13}{c}{$\blacktriangleright$
{\normalsize \boolea{}}~\cite{tarzan2025} ::
[Query: $\exists \Diamond( \mathit{ctr}_1 = 1 \land \mathit{ctr}_2 = 1 \land \dots \land \mathit{ctr}_K = 1 )$]
$\blacktriangleleft$} \\
\midrule
\rowcolor{white}
4  & $\varepsilon$ & 0.013 & 7.78 & 24
   & $\boldsymbol{\varepsilon}$ & \textbf{0.007} & \textbf{6.47} & \textbf{185}
   & 0.001 & 0.052 & 18.85 & 185 \\
\rowcolor{lightgray}
6  & \textbf{0.003} & \textbf{0.016} & \textbf{13.42} & \textbf{425}
   & 0.088 & 0.096 & 13.80 & 8502
   & 0.075 & 0.128 & 19.91 & 8502 \\
\rowcolor{white}
8  & \textbf{0.013} & \textbf{0.027} & \textbf{29.10} & \textbf{1009}
   & 75.461 & 75.980 & 220.82 & 566763
   & 77.204 & 77.523 & 81.79 & 566763 \\
\midrule
\multicolumn{13}{c}{$\blacktriangleright$
{\normalsize \flower{}}~\cite{J_rgensen_2012} :: 
[Query: $\exists \Diamond ( \mathit{Flower.Goal} )$]
$\blacktriangleleft$} \\
\midrule
\rowcolor{white}
5  & ${\boldsymbol \varepsilon}$ & \textbf{0.013} & \textbf{7.92} & \textbf{95}
   & 0.001 & 0.007 & 6.35 & 150
   & 0.001 & 0.056 & 18.82 & 150 \\
\rowcolor{lightgray}
7  & $\boldsymbol{\varepsilon}$ & \textbf{0.022} & \textbf{10.65} & \textbf{573}
   & 0.322 & 0.330 & 13.12 & 3508
   & 0.399 & 0.455 & 19.76 & 3508 \\
\rowcolor{white}
9  & \textbf{0.007} & \textbf{0.022} & \textbf{58.62} & \textbf{9161}
   & --- & TO & --- & ---
   & --- & TO & --- & --- \\
\midrule
\multicolumn{13}{c}{$\blacktriangleright$
{\normalsize \gates{}}~\cite{tarzan2025} ::
[Query: $\exists \Diamond ( \mathit{Unlocker.Goal} )$]
$\blacktriangleleft$} \\
\midrule
\rowcolor{white}
5  & $\boldsymbol{\varepsilon}$ & \textbf{0.013} & \textbf{8.84} & \textbf{76}
   & 0.001 & 0.008 & 7.19 & 549
   & 0.001 & 0.055 & 19.01 & 549 \\
\rowcolor{lightgray}
7  & $\boldsymbol{\varepsilon}$ & \textbf{0.014} & \textbf{11.23} & \textbf{151}
   & 0.110 & 0.122 & 21.08 & 32979
   & 0.035 & 0.092 & 22.79 & 32979 \\
\rowcolor{white}
9  & \textbf{0.001} & \textbf{0.015} & \textbf{15.43} & \textbf{251}
   & 47.895 & 48.771 & 806.72 & 3071646
   & 5.228 & 6.509 & 451.71 & 3071646 \\
\midrule
\multicolumn{13}{c}{$\blacktriangleright$
{\normalsize \ring{}}~\cite{tarzan2025} ::
[Query: $\exists \Diamond ( P_0.\mathit{Goal} \land P_1.\mathit{Goal} \land ... \land P_K.\mathit{Goal} )$]
$\blacktriangleleft$} \\
\midrule
\rowcolor{white}
4  & \textbf{0.001} & \textbf{0.015} & \textbf{12.40} & \textbf{542}
   & 0.020 & 0.029 & 11.56 & 7738
   & 0.010 & 0.065 & 19.52 & 7738 \\
\rowcolor{lightgray}
6  & \textbf{0.011} & \textbf{0.025} & \textbf{48.96} & \textbf{3397}
   & 12.137 & 12.561 & 268.93 & 1180575
   & 3.406 & 3.715 & 128.87 & 1180575 \\
\rowcolor{white}
8  & \textbf{0.204} & \textbf{0.225} & \textbf{687.77} & \textbf{47063}
   & --- & TO & --- & ---
   & --- & TO & --- & --- \\
\bottomrule
\end{tabular}
}
\end{table}
}
}

\textbf{Closed \ac{ta}}~Two networks of closed \acp{ta} were adapted from the benchmarks available on~\cite{tapaal2025}:
\medical{} models medical checks for patients, while \mpeg{} models the homonymous compression standard, where compression leverages a given number of bidirectional frames (BFrames).
Since \toolname{} does not currently support broadcast channels, these models were rewritten into equivalent versions that use only point-to-point channels.
The results in \autoref{tab:closed} reveal a high number of stored states.
This indicates that these models restrict the evolution of clocks in a way that forces zones to be repeatedly and finely split, thus pushing these models into a category where zone-based abstractions exhibit poor performance, a phenomenon already demonstrated in the above experiments (see \autoref{tab:punctual}).
Furthermore, both Uppaal and TChecker exhibit an unexpectedly high verification time in \medical{}. 
This indicates that the dominant cost in this benchmark originates from zone manipulation, whose operations may have quadratic or cubic complexity in the number of clocks~\cite{J_rgensen_2012}. 
Region-based analysis avoids this overhead, as its fundamental operations have linear complexity in the number of clocks.
Interestingly, TChecker aborted execution on the largest instances of \medical{} raising the error: \emph{"ERROR: value out of bounds"}.

\definecolor{lightgray}{gray}{0.95}
\setlength{\aboverulesep}{1pt}
\setlength{\belowrulesep}{1pt}
{\renewcommand{\arraystretch}{0.95}
{\setlength{\tabcolsep}{3.25pt}
\begin{table}[tb]
\centering
\caption{The results are derived from networks of closed \acp{ta}, where guards may either be punctual or contain non-strict inequalities.
In this setting, \toolname{} can handle both large networks (\medical{}) and large constants (\mpeg{}).}
\label{tab:closed}
\resizebox{\textwidth}{!}{
\begin{tabular}{%
    c|rrrr|rrrr|rrrr
}
\toprule
\multirow{2}{*}{\textbf{K}}
  & \multicolumn{4}{c|}{\textbf{TARZAN}}
  & \multicolumn{4}{c|}{\textbf{TChecker}}
  & \multicolumn{4}{c}{\textbf{UPPAAL}} \\
  & \multicolumn{1}{c}{\textbf{VT}}
  & \multicolumn{1}{c}{\textbf{ET}}
  & \multicolumn{1}{c}{\textbf{Mem}}
  & \multicolumn{1}{c|}{\textbf{Regions}}
  & \multicolumn{1}{c}{\textbf{VT}}
  & \multicolumn{1}{c}{\textbf{ET}}
  & \multicolumn{1}{c}{\textbf{Mem}}
  & \multicolumn{1}{c|}{\textbf{States}}
  & \multicolumn{1}{c}{\textbf{VT}}
  & \multicolumn{1}{c}{\textbf{ET}}
  & \multicolumn{1}{c}{\textbf{Mem}}
  & \multicolumn{1}{c}{\textbf{States}} \\
\midrule
\multicolumn{13}{c}{$\blacktriangleright$
{\normalsize \medical{}}~\cite{tapaal2025} ::
[Query: $\exists \Diamond ( \mathit{Patient}_1.\mathit{Done} \land \mathit{Patient}_2.\mathit{Done} \land \dots \land \mathit{Patient}_{K - 2}.\mathit{Done} ) $]
$\blacktriangleleft$} \\
\midrule
\rowcolor{white}
12  & 0.007 & 0.021 & 33.39 & 660
   & 0.011 & 0.030 & 22.63 & 639
   & \textbf{0.002} & \textbf{0.059} & \textbf{20.83} & \textbf{582} \\
\rowcolor{lightgray}
22  & 0.056 & 0.072 & 204.97 & 2715
   & 0.316 & 0.560 & 173.75 & 3774
   & \textbf{0.032} & \textbf{0.093} & \textbf{31.76} & \textbf{2862} \\
\rowcolor{white}
42  & \textbf{0.609} & \textbf{0.647} & \textbf{2271.65} & \textbf{15025}
   & 13.571 & 16.096 & 2517.18 & 25544
   & 1.275 & 1.379 & 230.07 & 16522 \\
\rowcolor{lightgray}
72  & \textbf{5.781} & \textbf{5.959} & \textbf{12891.17} & \textbf{68990}
   & 355.827 & 389.223 & 16109.45 & 126949
   & 54.200 & 54.838 & 2799.57 & 74762 \\
\rowcolor{white}
102  & \textbf{27.918} & \textbf{28.466} & \textbf{13780.57} & \textbf{189555}
   & KO & KO & KO & KO
   & 593.317 & 597.306 & 12364.67 & 202302 \\
\midrule
\multicolumn{13}{c}{$\blacktriangleright$
{\normalsize \mpeg{}}~\cite{tapaal2025} ::
[Query: $\exists \Diamond ( \mathit{BFrame}_1.\mathit{Bout} \land \mathit{BFrame}_2.\mathit{Bout} \land \dots \land \mathit{BFrame}_{K - 4}.\mathit{Bout} )$]
$\blacktriangleleft$} \\
\midrule
\rowcolor{white}
8  & 0.365 & 0.385 & 451.46 & 26916
   & 0.003 & 0.014 & 8.96 & 531
   & \textbf{0.001} & \textbf{0.058} & \textbf{19.46} & \textbf{537} \\
\rowcolor{lightgray}
12  & \textbf{0.796} & \textbf{0.819} & \textbf{718.63} & \textbf{27238}
   & 8.405 & 8.663 & 252.70 & 383103
   & 1.621 & 1.766 & 112.33 & 383109 \\
\rowcolor{white}
16  & \textbf{1.364} & \textbf{1.392} & \textbf{1032.46} & \textbf{27704}
   & --- & TO & --- & ---
   & --- & TO & --- & --- \\
\bottomrule
\end{tabular}
}
\end{table}
}
}

\textbf{Large constants}~When constants are large, zones are still preferable.
\autoref{tab:constants} reports results from two real-world case studies: the Carrier Sense Multiple Access / Collision Detection protocol (\csma{})~\cite{kiviriga2021randomized} and \pagerank{}~\cite{Baresi2020}, a mechanism used to prioritize Web keyword search results.
The latter is presented in two variants: in the first, the last two digits of the original constants were truncated, while the second uses the full constants.
The largest constant in \csma{} is 808, whereas in \pagerank{} they are 101 (truncated) and 10162 (full).

\definecolor{lightgray}{gray}{0.95}
\setlength{\aboverulesep}{1pt}
\setlength{\belowrulesep}{1pt}
{\renewcommand{\arraystretch}{0.95}
{\setlength{\tabcolsep}{3.25pt}
\begin{table}[tb]
\centering
\caption{Large constants negatively impact the \ac{vt} of regions, unlike zones.
Here, the first row of \pagerank{} corresponds to the variant with truncated constants.}
\label{tab:constants}
\resizebox{\textwidth}{!}{
\begin{tabular}{%
    c|rrrr|rrrr|rrrr
}
\toprule
\multirow{2}{*}{\textbf{K}}
  & \multicolumn{4}{c|}{\textbf{TARZAN}}
  & \multicolumn{4}{c|}{\textbf{TChecker}}
  & \multicolumn{4}{c}{\textbf{UPPAAL}} \\
  & \multicolumn{1}{c}{\textbf{VT}}
  & \multicolumn{1}{c}{\textbf{ET}}
  & \multicolumn{1}{c}{\textbf{Mem}}
  & \multicolumn{1}{c|}{\textbf{Regions}}
  & \multicolumn{1}{c}{\textbf{VT}}
  & \multicolumn{1}{c}{\textbf{ET}}
  & \multicolumn{1}{c}{\textbf{Mem}}
  & \multicolumn{1}{c|}{\textbf{States}}
  & \multicolumn{1}{c}{\textbf{VT}}
  & \multicolumn{1}{c}{\textbf{ET}}
  & \multicolumn{1}{c}{\textbf{Mem}}
  & \multicolumn{1}{c}{\textbf{States}} \\
\midrule
\multicolumn{13}{c}{$\blacktriangleright$
{\normalsize \csma{}}~\cite{kiviriga2021randomized} ::
[Query: $\exists \Diamond \left( \parbox{7.5cm}{\begin{math}P_1.\mathit{senderRetry} \land P_2.\mathit{senderRetry} \land P_3.\mathit{senderTransm} \end{math} \\ \begin{math}\land \; P_4.\mathit{senderRetry} \land \dots \land P_7.\mathit{senderRetry}\end{math}} \right)$]
$\blacktriangleleft$} \\
\midrule
\rowcolor{white}
21  & 0.198 & 0.215 & 364.96 & 20936
   & 0.004 & 0.038 & 30.32 & 232
   & \textbf{0.001} & \textbf{0.061} & \textbf{20.24} & \textbf{269} \\
\rowcolor{lightgray}
31  & 0.674 & 0.701 & 1194.45 & 47526
   & 0.018 & 0.218 & 117.03 & 371
   & \textbf{0.002} & \textbf{0.064} & \textbf{21.11} & \textbf{554} \\
\rowcolor{white}
51  & 3.364 & 3.431 & 5351.35 & 131906
   & 0.293 & 1.386 & 837.12 & 1381
   & \textbf{0.010} & \textbf{0.077} & \textbf{23.89} & \textbf{1424} \\
\midrule
\multicolumn{13}{c}{$\blacktriangleright$
{\normalsize \pagerank{}}~\cite{Baresi2020} ::
[Query: $\exists \Diamond ( \mathit{Stage}_7.\mathit{Completed} )$]
$\blacktriangleleft$} \\
\midrule
\rowcolor{white}
9  & 0.068 & 0.085 & 179.33 & 4193
   & 0.002 & 0.086 & 58.54 & 1129
   & \textbf{0.002} & \textbf{0.076} & \textbf{22.89} & \textbf{1129} \\
\rowcolor{lightgray}
9  & 4.227 & 4.361 & 10950.36 & 257311
   & 0.002 & 0.085 & 58.53 & 1129
   & \textbf{0.002} & \textbf{0.076} & \textbf{23.01} & \textbf{1129} \\
\bottomrule
\end{tabular}
}
\end{table}
}
}

\textbf{Safety properties}~When safety properties must be verified using forward reachability, \toolname{} requires to explore the entire reachable state space.
\autoref{tab:safety} reports the results for the Fischer mutual exclusion protocol~\cite{uppaalSite}.
The verified property specifies that two given processes must never be in the critical section simultaneously.
\Red{In the reminder of this section, we demonstrate that backward reachability may enhance the verification of safety properties by exploring only a restricted subset of the state space, potentially avoiding out-of-memory events.}

\definecolor{lightgray}{gray}{0.95}
\setlength{\aboverulesep}{1pt}
\setlength{\belowrulesep}{1pt}
{\renewcommand{\arraystretch}{0.95}
{\setlength{\tabcolsep}{3.25pt}
\begin{table}[t]
\centering
\caption{Safety properties require \toolname{} to visit the entire reachable state space.}
\label{tab:safety}
\resizebox{\textwidth}{!}{
\begin{tabular}{%
    c|rrrr|rrrr|rrrr
}
\toprule
\multirow{2}{*}{\textbf{K}}
  & \multicolumn{4}{c|}{\textbf{TARZAN}}
  & \multicolumn{4}{c|}{\textbf{TChecker}}
  & \multicolumn{4}{c}{\textbf{UPPAAL}} \\
  & \multicolumn{1}{c}{\textbf{VT}}
  & \multicolumn{1}{c}{\textbf{ET}}
  & \multicolumn{1}{c}{\textbf{Mem}}
  & \multicolumn{1}{c|}{\textbf{Regions}}
  & \multicolumn{1}{c}{\textbf{VT}}
  & \multicolumn{1}{c}{\textbf{ET}}
  & \multicolumn{1}{c}{\textbf{Mem}}
  & \multicolumn{1}{c|}{\textbf{States}}
  & \multicolumn{1}{c}{\textbf{VT}}
  & \multicolumn{1}{c}{\textbf{ET}}
  & \multicolumn{1}{c}{\textbf{Mem}}
  & \multicolumn{1}{c}{\textbf{States}} \\
\midrule
\multicolumn{13}{c}{$\blacktriangleright$
{\normalsize \fischer{}}~\cite{uppaalSite} ::
[Query: $\exists \Diamond ( \mathit{Fischer}_1.\mathit{cs} \land \mathit{Fischer}_2.\mathit{cs} ) $]
$\blacktriangleleft$} \\
\midrule
\rowcolor{white}
4  & 0.644 & 0.669 & 1098.36 & 180032
   & $\boldsymbol{\varepsilon}$ & \textbf{0.007} & \textbf{6.75} & \textbf{220}
   & 0.001 & 0.055 & 18.96 & 220 \\
\rowcolor{lightgray}
5  & 19.627 & 19.849 & 12948.99 & 3920652
   & 0.004 & 0.011 & 9.30 & 727
   & \textbf{0.002} & \textbf{0.057} & \textbf{19.07} & \textbf{727} \\
\rowcolor{white}
6  & --- & --- & OOM & ---
   & 0.024 & 0.032 & 12.25 & 2378
   & \textbf{0.008} & \textbf{0.064} & \textbf{19.43} & \textbf{2378} \\
\bottomrule
\end{tabular}
}
\end{table}
}
}

\subsection{Backward reachability analysis with \textnormal{\toolname{}}}
We begin with a backward unreachability example on \flower{}, and then compare its forward and backward reachability results.
The entire \flower{} reachable state space, when considering five clocks and computed going forward starting from the initial region
$
\reg_0 = \{ q_0, h(x_1) = h(x_2) = h(x_3) = h(x_4) = h(y) = 0, X_0 \}$, where $X_{\reg_0, 0} = \{ x_1, x_2, x_3, x_4, y \}
$, resulted in 1517 regions, computed in 1.18 ms.

Backward reachability can offer a practical advantage for proving unreachability, since it may avoid a complete state space exploration.
Consider the region
$
\reg_1 = \{ q_0, h(x_1) = 1, h(x_2) = 2, h(x_3) = h(x_4) = 0, h(y) = 1, X_{-1}, X_0 \}$, where 
$X_{\reg_1, -1} = \{ x_2, y \}$ and $X_{\reg_1, 0} = \{ x_1, x_3, x_4 \}
$.
$\reg_1$ is unreachable: it requires both $x_2$ and $y$ to become unbounded simultaneously, but the constraint $y \ge 1$ forces $y$ to become unbounded before $x_2$.
\toolname{} classifies $\reg_1$ as unreachable by verifying that $\reg_0$ is unreachable from $\reg_1$, requiring 272 regions in 0.245 ms.
By adjusting $\reg_1$ with the correct ordering in which clocks become unbounded, we obtain
$
\reg_2 = \{ q_0, h(x_1) = 1, h(x_2) = 2, h(x_3) = h(x_4) = 0, h(y) = 1, X_{-2}, X_{-1}, X_0 \}$, where 
$X_{\reg_2, -2} = \{ x_2 \}$, $X_{\reg_2, -1} = \{ y \}$, and $X_{\reg_2, 0} = \{ x_1, x_3, x_4 \}
$.
Region $\reg_2$ is reachable.
Reaching $\reg_0$ from $\reg_2$ required 263 regions in 0.378 ms.

Now, consider the version of \flower{} with $K = n + 1$ clocks.  
Let $\reg_K$ denote the region in location $\mathit{Goal}$ in which all clocks $x_i$ are exactly equal to zero (no fractional part) and $y \ge 1$ holds. 
\autoref{tab:backward_scalability} reports the \ac{vt} and number of regions stored by \toolname{} when checking the forward and backward reachability of $\reg_K$.
The results show that forward reachability is generally faster and requires fewer regions, while backward reachability remains comparable \Red{in this setting}.

\begin{table}[tb]
\centering
\caption{Backward reachability results. \fischer{} goes out of memory at $K=6$. Tests use \ac{dfs} except \train{}, for which \ac{bfs} is faster. \flower{} checks reachability, \fischer{} unreachability, and \train{} reachability (rows 1,2) and unreachability (rows 3,4).}
\label{tab:backward_scalability}
{\setlength{\tabcolsep}{2.25pt}
\resizebox{\textwidth}{!}{
\begin{tabular}{r r | r r !{\vrule width 1pt} r r | r r !{\vrule width 1pt} r r | r r}
\toprule
\multicolumn{4}{c!{\vrule width 1pt}}
{$\blacktriangleright$ {\normalsize {\flower{}}~\cite{J_rgensen_2012}} :: DFS $\blacktriangleleft$} & 
\multicolumn{4}{c!{\vrule width 1pt}}
{$\blacktriangleright$ {\normalsize {\fischer{}}~\cite{tarzan2025}} :: DFS $\blacktriangleleft$} & 
\multicolumn{4}{c}
{$\blacktriangleright$ {\normalsize{\train{}}~\cite{tarzan2025}} :: BFS $\blacktriangleleft$} \\
\midrule
\multicolumn{2}{c|}{\makecell{\emph{Forward}\\$K = 11,13,15,17$}}
& \multicolumn{2}{c!{\vrule width 1pt}}{\makecell{\emph{Backward}\\$K = 11,13,15,17$}}
& \multicolumn{2}{c|}{\makecell{\emph{Forward}\\$K = 3,4,5,6$}}
& \multicolumn{2}{c!{\vrule width 1pt}}{\makecell{\emph{Backward}\\$K = 3,4,5,6$}}
& \multicolumn{2}{c|}{\makecell{\emph{Forward}\\$K = 4,5,4,5$}}
& \multicolumn{2}{c}{\makecell{\emph{Backward}\\$K = 4,5,4,5$}} \\
\multicolumn{1}{c}{\textbf{VT}} 
& \multicolumn{1}{c|}{\textbf{Regions}}
& \multicolumn{1}{c}{\textbf{VT}}
& \multicolumn{1}{c!{\vrule width 1pt}}{\textbf{Regions}}
& \multicolumn{1}{c}{\textbf{VT}}
& \multicolumn{1}{c|}{\textbf{Regions}}
& \multicolumn{1}{c}{\textbf{VT}}
& \multicolumn{1}{c!{\vrule width 1pt}}{\textbf{Regions}} 
& \multicolumn{1}{c}{\textbf{VT}}
& \multicolumn{1}{c|}{\textbf{Regions}}
& \multicolumn{1}{c}{\textbf{VT}} 
& \multicolumn{1}{c}{\textbf{Regions}} \\
\midrule
\rowcolor{white} 0.028 & 30761 & 0.060 & 38669 & 0.007 & 9083 & $\varepsilon$ & 264 & $\varepsilon$ & 1241 & 0.171 & 118868 \\
\rowcolor{lightgray} 0.412 & 370331 & 0.615 & 388911 & 0.141 & 183459 & 0.001 & 1316 & 0.026 & 31516 & 66.877 & 31051667 \\

\hhline{~~~~~~>{\arrayrulecolor{black}}~~---->{\arrayrulecolor{black}}}

\rowcolor{white} 7.188 & 5142671 & 9.414 & 4921043 & 4.061 & 4182681 & 0.008 & 8724 & 0.013 & 17124 & 0.030 & 26914 \\
\rowcolor{lightgray} 16.212 & 11023829 & 21.497 & 11077295 & --- & --- & 0.100 & 77084 & 1.000 & 912597 & 0.635 & 614220 \\
\bottomrule
\end{tabular}%
}
}
\end{table}

Backward reachability was employed to verify that in \fischer{} two processes cannot occupy the critical section simultaneously.
As backward reachability does not currently support networks of \ac{ta} nor integer variables, we developed a script to construct a single \ac{ta} version of \fischer{}~\cite{tarzan2025}, in which the \emph{pid} variable is directly encoded within its locations.
The safety property to be verified with $K$ processes is described as follows: processes $1$ through $K-2$ are requesting access to the critical section, with their clocks strictly between $1$ and $2$; processes $K-1$ and $K$ are both in the critical section, with their clocks greater than $2$; \emph{pid} is equal to $K$.
A preliminary phase was necessary to compute the initial subset of regions from which backward reachability was started. 
This required considering all ordered partitions of the fractional parts of bounded clocks as well as all ordered partitions of unbounded clocks.
\Red{Clock constraints in the property were introduced to enable the manual computation of this subset.}
Backward reachability results, reported in \autoref{tab:backward_scalability}, demonstrate superior performance compared to forward reachability, \Red{to the extent that an out-of-memory event becomes feasible}.
\Red{Indeed, the unreachable portion of the state space limited the backward exploration to a relatively small number of regions.}
We evaluated the same property with Uppaal (which uses forward reachability) on the same \fischer{} model and obtained:
for $K=3$, 88 states in 0.001 s;
for $K=4$, 460 states in 0.004 s;
for $K=5$, 2765 states in 0.025 s;
and for $K=6$, 19042 states in 0.214 s.
In this comparison, \toolname{} outperformed Uppaal when using backward reachability.

Backward reachability was also applied to \train{}~\cite{tarzan2025} (we developed a script to convert it into a single \ac{ta} version).
Exploration was conducted in \ac{bfs}, as it achieved better performance on this model.
\train{} allows two or more trains to be in the crossing while the gate is down.
Rows 1,2 of \autoref{tab:backward_scalability} report the results for (i) reaching the first region detected by the \ac{bfs} exploration in which at least two trains occupy the crossing while the gate is down, and (ii) computing the backward predecessors from this region to an initial one. 
In this case, backward reachability requires substantially more regions than forward reachability.
We also verified the safety property stating that whenever a train is in the crossing, the number of trains in the crossing cannot be zero (\emph{i.e.}, no counter malfunction).
Such property with $K-2$ trains, a gate, and a controller, is described as follows: trains $1$ to $K-3$ are not in the crossing and their clocks are strictly between $2$ and $3$; train $K-2$ is in the crossing and its clock is exactly $4$; the gate's clock is strictly between $1$ and $2$, and the controller's clock is greater than $1$. 
As in \fischer{}, we first explicitly identified the initial subset of regions from which backward reachability was started. 
Rows 3,4 of \autoref{tab:backward_scalability} show that, in this case, forward and backward reachability perform comparably.

\section{Conclusion and Future Works}
\label{sec:conclusion}
We introduced \toolname{}, a region-based verification library supporting forward and backward reachability for \ac{ta}. 
Our experiments demonstrated that \toolname{} exhibits superior performance on closed \ac{ta} and \ac{ta} with punctual guards, while zone-based approaches remain better suited to other classes of \ac{ta}. 
We showed that backward reachability with \toolname{} is often feasible in practice and effective for verifying safety properties.
\Red{For this reason, the integration of \toolname{} into existing tools has the potential to enhance their verification capabilities.}

Future work will proceed in several directions. 
We plan to integrate additional optimizations into \toolname{}, with a particular focus on improving the computation of delay successors and predecessors, while leveraging symbolic techniques to reduce memory usage. 
In addition, we will extend backward reachability algorithms to networks of \ac{ta}, thereby laying the foundation for applying \toolname{} to the synthesis of winning strategies in Timed Games.
Finally, we intend to explore hybrid approaches that complement \toolname{} with zone-based techniques, thereby combining the strengths of both abstractions.

\bibliographystyle{splncs04}
\bibliography{bibliography}

\iftoggle{EXTENDED_VERSION}{
\newpage
\appendix
\setcounter{theorem}{0}
\setcounter{lemma}{0}

\section{Immediate Delay Successors}
\label{app:foreword_delay_successors_algorithm}
\autoref{alg:finddelaysucc}, whose time complexity is $O(\lvert \ackc \rvert)$, computes the (unique) immediate delay successor of a given region $\reg$, based on the class to which $\reg$ belongs.  
In particular, Line~\ref{line:fds:copyr} creates a copy $\breg$ of $\reg$, returned as the result in Line~\ref{line:fds:returnreg}.  
If $\reg$ belongs to class $\cu$ (Line~\ref{line:fds:isofclassu}), it has no delay successors at all, so we return the same region on Line~\ref{line:fds:returnsame}.
If $\reg$ belongs to either class $\ca$ or class $\cc$ (Line~\ref{line:fds:isofclassac}), two auxiliary sets are initialized on Line~\ref{line:fds:inittmpoob}:  
$X_{\mathit{tmp}}$ stores the clocks that exit the unit (\emph{i.e.}, leaving $X_{\breg, 0}$), while $X_{\mathit{oob}}$ collects the clocks that become unbounded.  
For each clock $x \in X_{\breg, 0}$ (with no fractional part), the loop at Line~\ref{line:fds:foreachxinx0} adds $x$ to either $X_{\mathit{oob}}$ (Line~\ref{line:fds:insertintooob}) or $X_{\mathit{tmp}}$ (Line~\ref{line:fds:insertinttmp}), and removes $x$ from $X_{\breg, 0}$ (Line~\ref{line:fds:removexfromx0}).  
If any clocks left the unit, they are inserted to the right of $X_{\breg, 0}$, as they now have the smallest fractional part (Lines~\ref{line:fds:tmpisnotempty},\ref{line:fds:insertxtmpright}).  
If a clock becomes unbounded, it is placed to the left of $X_{\breg, -\lidx_{\breg}}$, reflecting that it is the most recently unbounded clock (Lines~\ref{line:fds:oobisnotempty},\ref{line:fds:insertxoobright}).  
If $\reg$ is of class $\cb$ (Line~\ref{line:fds:isofclassb}), then the clocks with the highest fractional part (\emph{i.e.}, those in $X_{\breg, \ridx_{\breg}}$) reach the next unit.  
The loop at Line~\ref{line:fds:incrementx} increments their integer value by 1 (Line~\ref{line:fds:incrementxbodyloop}).  
Since these clocks have entered the next unit, they are inserted into $X_{\breg, 0}$, and $X_{\breg, \ridx_{\breg}}$ is removed from $\breg$ (Line~\ref{line:fds:insertremove}).

\begin{algorithm}[th]
    \caption{\texttt{find-immediate-delay-successor}$(\reg)$}
    \label{alg:finddelaysucc}
    
    \KwData{$\reg$: a region.}
    \KwResult{The immediate delay successor of $\reg$.}
    \Begin() {
        $\breg \gets$ copy of $\reg$\; \label{line:fds:copyr}
        \uIf{$X_{\reg,0} = \varnothing \land \ridx_{\reg} = 0$}{ \label{line:fds:isofclassu}
            \tcp{$\reg$ belongs to class U}
            \Return $\reg$\; \label{line:fds:returnsame}
        }
        \uElseIf{$X_{\reg,0} \neq \varnothing$}{ \label{line:fds:isofclassac}
            \tcp{$\reg$ belongs to either class Z or class M}
            $X_{\mathit{tmp}} \gets \varnothing$; $X_{\mathit{oob}} \gets \varnothing$\; \label{line:fds:inittmpoob}
            
            \ForEach{$x \in X_{\breg,0}$}{ \label{line:fds:foreachxinx0}
                \uIf{$h_{\breg}(x) = c_{\max}$}{ \label{line:fds:insertintooobif}
                    $X_\mathit{oob} \gets X_\mathit{oob} \cup \{x\}$\; \label{line:fds:insertintooob}
                }
                \Else{
                    $X_{\mathit{tmp}} \gets X_{\mathit{tmp}} \cup \{x\}$\; \label{line:fds:insertinttmp}
                }
                $X_{\breg,0} \gets X_{\breg,0} \setminus \{x\}$\; \label{line:fds:removexfromx0}
            }
            
            \If{$X_{\mathit{tmp}} \neq \varnothing$}{ \label{line:fds:tmpisnotempty}
                insert $X_{\mathit{tmp}}$ to the right of $X_{\breg, 0}$\; \label{line:fds:insertxtmpright}
            }
            \If{$X_{\mathit{oob}} \neq \varnothing$}{ \label{line:fds:oobisnotempty}
                insert $X_{\mathit{oob}}$ to the left of $X_{\breg, -\lidx_{\breg}}$\; \label{line:fds:insertxoobright}
            }
            }
        \Else{ \label{line:fds:isofclassb}
         \tcp{$\reg$ belongs to class P}
                \ForEach{$x \in X_{\breg, \ridx_{\breg}}$}{ \label{line:fds:incrementx}
                    $h_{\breg}(x) \gets h_{\breg}(x) + 1$\; \label{line:fds:incrementxbodyloop}
                }
                insert clocks of $X_{\breg, \ridx_{\breg}}$ in $X_{\breg, 0}$ and remove $X_{\breg, \ridx_{\breg}}$ from $\breg$\; \label{line:fds:insertremove}
            }
        
        \Return $\breg$\; \label{line:fds:returnreg}
        }
\end{algorithm}

\section{Discrete predecessors}
\label{app:discrete_predecessors_algorithm}
\begin{algorithm}[t]
    \caption{\texttt{find-discrete-predecessors}$(\reg)$}
    \label{alg:finddiscretepred}
    
    \KwData{$\reg$: a region.}
    \KwResult{The discrete predecessors of $\reg$.}
    
    \Begin(){
        $\mathit{res} \gets \varnothing$\; \label{line:fdp:resempty}
        
        \ForEach{$t  = (q, a, \gamma, Y, q_{\reg})  \in  T$}{ \label{line:fdp:startloop}
        
        $X_{\mathit{bnd}} \gets \varnothing$; $X_{\mathit{oob}} \gets \varnothing$\; \label{line:fdp:initbndoob}
        $\breg \gets$ copy of $\reg$\; \label{line:fdp:copyr}
        $\allclocks{\breg} \gets$ a reference to the clocks of $\breg$\; \label{line:fdp:clockreference}
        $q_{\breg} \gets q$\; \label{line:fdp:assignq}

        \ForEach{clock $x \in Y$} { \label{line:fdp:foreachcconstr}

            \uIf{$x$ belongs to a constraint of the form $(x = c) \in \gamma$}{ \label{line:fdp:isequality1}
                $h_{\breg}(x) = c$\; \label{line:fdp:isequality2}
            }

            \uElseIf{$x$ belongs to a constraint of the form $(x > \cmax) \in \gamma$}{ \label{line:fdp:isunbo1}
                $h_{\breg}(x) = \cmax$\; \label{line:fdp:assignqoob}
                $X_{\mathit{oob}} \gets X_{\mathit{oob}} \cup \{ x \}$;
                $\allclocks{\breg} \gets \allclocks{\breg} \setminus \{ x \}$\; \label{line:fdp:isunbo2}
            }

            \Else{ \label{line:fdp:isbound1}
                $X_{\mathit{bnd}} \gets X_{\mathit{bnd}} \cup \{ x \}$;
                $\allclocks{\breg} \gets \allclocks{\breg} \setminus \{ x \}$\;  \label{line:fdp:isbound2}
            }
        }

        \uIf{$X_{\mathit{bnd}} = \varnothing$}{ \label{line:fdp:bndisempty}
            $\res \gets \res$ $\cup$ \texttt{part-regs}($\breg, -\lidx_{\breg}, -1, X_\mathit{oob}, h_{\breg}$)\; \label{line:fdp:oobfubini}
        }
        \Else{

        \ForEach{function $\bar{h} : X_{\mathit{bnd}} \to \{0, \dots, \cmax + 1\}$}{ \label{line:fdp:foreachfunc}
           
            \label{line:fdp:gammasatisfiedbyh}
                
                $\Delta \gets \{ x \in X_{\mathit{bnd}} \; \vert \; \bar{h}(x) > \cmax \}$\; \label{line:fdp:createdelta}
                
                \ForEach{$R \in$ \textnormal{\texttt{part-regs}}$(\breg, -\lidx_{\breg}, -1, X_\mathit{oob} \cup \Delta, \bar{h} \cup h_{\breg})$}{ \label{line:fdp:allrinpermreg}
                    $\res \gets \res$ $\cup$ \texttt{part-regs}($R, 0, \ridx_R, X_{\mathit{bnd}} \setminus \Delta, \bar{h} \cup h_{R}$)\; \label{line:fdp:resUres}
                    }
            }
        }   
        remove from $\res$ all regions not satisfying $\gamma$\; \label{line:fdp:removefromres}
        }

        \Return $\res$\; \label{line:fdp:returnres}
}
\end{algorithm}

\begin{algorithm}[th]
    \caption{\texttt{part-regs}$(\reg, i, j, X, \mathcal{H})$}
    \label{alg:perm-regs}
    
    \KwData{$\reg$: a region; 
    $i, j$: two integers such that $i \le j$ holds;
    $X$: a set of clocks,
    $\mathcal{H}$: a function assigning integer values to clocks.}
    \KwResult{all regions obtained from $\reg$ by finding all ordered partitions of $X$ while preserving the relative order of $X_{\reg,i}, \dots, X_{\reg,j}$.}
    \Begin(){
            $\res \gets \varnothing$\; \label{line:fdppr:res}
            $\breg \gets $ copy of $\reg$\; \label{line:fdppr:copyreg}

            \If{$X = \varnothing$}{ \label{line:fdppr:xisempty}
                \Return $\{ \reg \}$\; \label{line:fdppr:returnbregxempty}
            }

            \ForEach{$x \in X$}{ \label{line:fdppr:foreachx}
                \If{$i \ge 0 \land \mathcal{H}(x) = \cmax$}{ \label{line:fdppr:iandhxcond}
                    $X_{\breg, 0} \gets X_{\breg, 0} \cup \{ x \}$;
                    $X \gets X \setminus \{ x \}$\; \label{line:fdppr:removeinsertx}
                }
            }
            $\mathcal{X} \gets \{ X_{\breg, i}, \dots, X_{\breg, j} \}$\; \label{line:fdppr:preserve}
            $\Pi \gets $ all partitions of $X$\; \label{line:fdppr:partitions}

            \ForEach{partition $\pi \in \Pi$}{ \label{line:fdppr:forpartitions}
                \begin{varwidth}[t]{0.8\linewidth}{$\mathcal{Z} \gets$ ordered partitions of $(\mathcal{X} \cup \pi)$ that preserve the relative order of the elements in $\mathcal{X}$; these partitions are obtained by computing all permutations of $\pi$ and interleaving them with $\mathcal{X}$, allowing elements of $\pi$ to also be inserted within the elements of $\mathcal{X}$}\; \label{line:fdppr:orderedpartitions}
                \end{varwidth}
                \nl
                \ForEach{ordered partition $\zeta \in \mathcal{Z}$}{ \label{line:fdppr:foreachpermut}
                   \begin{varwidth}[t]{0.8\linewidth}{$R \gets $ copy of $\breg$ such that ($\mathcal{X} \cup \pi$) are partitioned according to $\zeta$ and clocks have a value given by $\mathcal{H}$ (if $\mathcal{H}(x) > \cmax$, then assign $\cmax$ to $x$)\;} \label{line:fdppr:computecopy}
                   \end{varwidth}
                   \nl
                   $\res \gets \res \cup \{ R \}$\; \label{line:fdppr:collectr}
                }
           }
           \Return $\res$\; \label{line:fdppr:returnres}
        }        
\end{algorithm}

The discrete predecessors of a region $\reg$ can be computed using \autoref{alg:finddiscretepred}.
In particular, Line~\ref{line:fdp:resempty} initializes a set storing the computed predecessor regions.
A loop over all transitions $t$ ending in $q_{\reg}$, the location of $\reg$, 
begins at Line~\ref{line:fdp:startloop}.
Line~\ref{line:fdp:initbndoob} initializes a set $X_{\mathit{bnd}}$, which will contain clocks
that were bounded or possibly unbounded, and a set $X_{\mathit{oob}}$, which will contain clocks that were certainly unbounded (in a predecessor).
A copy $\breg$ of $\reg$ is made at Line~\ref{line:fdp:copyr}.
Line~\ref{line:fdp:clockreference} creates a reference $\allclocks{\breg}$ to the clock sets of $\breg$, meaning that changes to $\allclocks{\breg}$ (\emph{e.g.}, removing a clock) directly affect $\breg$.
The location of $\breg$ is updated on Line~\ref{line:fdp:assignq} to match the one where $t$ originates.
\Red{Then, the loop at Line~\ref{line:fdp:foreachcconstr} examines each reset clock $x \in Y$:
(i) if $x$ appears in an equality clock constraint of the form $(x = c) \in \gamma$, only the integer value of $x$ must be set (Lines~\ref{line:fdp:isequality1},\ref{line:fdp:isequality2});
(ii) if $x$ was unbounded, \emph{i.e.}, it appears in an inequality constraint of the form $(x > \cmax) \in \gamma$, $x$ is added to $X_{\mathit{oob}}$ and removed from $\breg$ (Lines~\ref{line:fdp:isunbo1},\ref{line:fdp:isunbo2}), while its integer value is set to $\cmax$ (Line~\ref{line:fdp:assignqoob});
(iii) otherwise, $x$ was bounded or possibly unbounded, hence
$x$ is added to $X_{\mathit{bnd}}$ and removed from $\breg$ (Lines~\ref{line:fdp:isbound1},\ref{line:fdp:isbound2}).
The latter case also accounts for reset clocks $x' \in Y$ that do not appear in any constraint of $\gamma$, which is equivalent to assuming a trivial constraint of the form $(x' \geq 0) \in \gamma$.
It is important to note that every clock in $Y$ is required to be exactly zero, \emph{i.e.}, to have an integer value equal to zero and no fractional part, since all clocks in $Y$ are reset along transition $t$.}

If $X_{\mathit{bnd}}$ is empty (Line~\ref{line:fdp:bndisempty}), the predecessors over the current transition are computed by finding all ordered partitions of $X_{\mathit{oob}}$, preserving the order of unreset unbounded clocks. This is performed at Line~\ref{line:fdp:oobfubini} using the auxiliary algorithm \emph{part-regs}$(\reg', i, j, X, \mathcal{H})$,
which computes all regions derived from the input region $\reg'$ by finding all ordered partitions of a clock set $X$, while preserving the order of \Red{clock sets} $X_{\reg',i}, \dots, X_{\reg',j}$ (where $i \leq j$).
For each resulting region, the integer value of every clock $x$ is given by $\mathcal{H}(x)$.
\Red{If $X = \varnothing$, \emph{part-regs} returns the input region $\reg'$ unchanged.  
This also covers the case where a transition does not reset any clock; in such a situation, both $X_{\mathit{bnd}}$ and $X_{\mathit{oob}}$ are empty, and Line~\ref{line:fdp:oobfubini} simply inserts $\breg$ into the result set.}
The pseudocode of \emph{part-regs} is given in \autoref{alg:perm-regs}, and a detailed description follows at the end of this section.

Otherwise, all combinations of integer values (and potential unboundedness) for clocks in $X_{\mathit{bnd}}$ must be considered. This is done in the loop at Line~\ref{line:fdp:foreachfunc}, which iterates over all functions $\bar{h}$ defined on $X_{\mathit{bnd}}$. Note that $\bar{h}$ maps only clocks in $X_{\mathit{bnd}}$, while $h_{\breg}$ maps clocks in $\allclocks{\breg} \setminus X_{\mathit{bnd}}$; the conventional value $\cmax + 1$ captures the case in which a clock was unbounded in the predecessor.
All clocks in $X_{\mathit{bnd}}$ that are unbounded under $\bar{h}$ ({\emph{i.e.}, the integer value of these clocks returned by $\bar{h}$ is $\cmax$ + 1}) are collected into a set $\Delta$ (Line~\ref{line:fdp:createdelta}).
The loop at Line~\ref{line:fdp:allrinpermreg} iterates over all regions generated by \emph{part-regs} from $\breg$ partitioning $X_{\mathit{oob}} \cup \Delta$, preserving the order of unreset unbounded clocks.
For each such region $R$, \emph{part-regs} is applied on Line~\ref{line:fdp:resUres} to $R$ partitioning $X_{\mathit{bnd}} \setminus \Delta$ (\emph{i.e.}, the bounded clocks), preserving the order of unreset bounded clocks.
\Red{On Line~\ref{line:fdp:removefromres}, all regions not satisfying the guard $\gamma$ of $t$ are removed from $\mathit{res}$.}
Line~\ref{line:fdp:returnres} returns the discrete predecessors of $\reg$.

\Red{It is worth noting that the total number of functions $\bar{h}$ defined over $X_{\mathit{bnd}}$ in Line~\ref{line:fdp:foreachfunc} can be reduced by tracking the minimum and maximum integer values that the clocks may assume. 
For example, consider two clocks $x,y$ appearing in the constraints $1 \leq x \land x \leq 4 \land y \geq 5$. 
In this case, it is unnecessary to generate all combinations of integer values for $x$ and $y$ from $0$ to $\cmax + 1$; instead, it suffices to generate combinations where $\bar{h}(x) \in [1,4]$ and $\bar{h}(y) \in [5,\cmax+1]$.
In addition, note that \emph{part-regs} may generate invalid regions that must be discarded. 
This occurs, for instance, when a clock $x$ appears in a constraint $x > c$ and is assigned the integer value $\bar{h}(x) = c$. 
In this case, $x$ must have a fractional part greater than zero; however, when computing all possible ordered partitions, it may happen that $x$ is assigned no fractional part. 
Similarly, if $x$ appears in a constraint $x \le c$ and is assigned $\bar{h}(x) = c$, it must have a fractional part equal to zero, but the computation of ordered partitions may assign it a fractional part greater than zero.  
These incorrect edge cases can be handled by adding additional checks to Algorithm~\ref{alg:finddiscretepred} and \emph{part-regs}, which are omitted here for conciseness.}

\textbf{Description of \emph{part-regs}}~Line~\ref{line:fdppr:res} initializes the result set.
Then, a copy $\breg$ of $\reg$ is made on Line~\ref{line:fdppr:copyreg} to avoid modifying the original input region.
Notice that, if the set $X$ is empty (Line~\ref{line:fdppr:xisempty}), $X$ has exactly one partition, \emph{i.e.}, the empty set itself. 
Therefore, for convenience, the original region is returned on Line~\ref{line:fdppr:returnbregxempty}.
The loop at Line~\ref{line:fdppr:foreachx} processes each clock $x \in X$, checking if it can be inserted into $X_{\breg, 0}$ (Line~\ref{line:fdppr:iandhxcond}) and doing so at Line~\ref{line:fdppr:removeinsertx}, while also removing $x$ from $X$.
This is done since a bounded clock can have an integer value equal to $\cmax$ only when it has no fractional part.
Lines~\ref{line:fdppr:preserve} and \ref{line:fdppr:partitions} compute the sets for which the order must be preserved and the partitions $\Pi$ of $X$, respectively.

The loop at Line~\ref{line:fdppr:forpartitions} iterates over $\Pi$, generating ordered partitions preserving the order of $\mathcal{X}$ (\emph{e.g.}, in a valid ordered partition involving $\mathcal{X} = \{X_1, X_2, X_3\}$, $X_1$ appears before $X_2$, which appears before $X_3$).
Ordered partitions are computed at Line~\ref{line:fdppr:orderedpartitions}; some of them may involve inserting elements of a given partition $\pi \in \Pi$ into elements of $\mathcal{X}$.
For each valid ordered partition $\zeta$ (Line~\ref{line:fdppr:foreachpermut}), a region is created by copying $\breg$ (Line~\ref{line:fdppr:computecopy}), updating its sets of clocks according to $\zeta$, and using $\mathcal{H}$ for assigning integer values to clocks (the values returned by $\mathcal{H}$ are clamped to ensure no value exceeds $\cmax$).
Each such region is added to the result set (Line~\ref{line:fdppr:collectr}), which is returned on Line~\ref{line:fdppr:returnres}.

\begin{example}
    Let $\aut$ be a \ac{ta} such that $\ackc = \{ x, y, w, p, z, s \}$ and $\cmax = 5$, and let $t = (q_0, a, \gamma, Y, q) \in T$, where $\gamma :=x \ge 0 \land y \ge 1 \land  w > 5 \land p = 1 \land z > 3 \land s > 4$ and $Y = \{ x, y, w, p \}$, be the transition of $\aut$ over which, for illustrative purposes, we compute only one discrete predecessor, which we refer to as $\breg$, of the following region:
    $
    \reg = \{ q, h(x) = h(y) = h(w) = h(p) = 0, h(z) = h(s) = 4, X_0, X_1, X_2 \}
    $,
    where $X_0 = \{x, y, w, p \}$, $X_1 = \{ z \}$, and $ X_2 = \{ s \}$.
    To better understand how \autoref{alg:finddiscretepred} works, we reference some of its lines in the following computation.

    Since clocks $z$ and $s$ are not reset in $t$, they do not influence the computation of $\breg$.  
    From the analysis of $Y$ (Lines~\ref{line:fdp:foreachcconstr} through~\ref{line:fdp:isbound2}), we can infer the following: $x$ and $y$ may have been either bounded or unbounded in $\breg$ (here, we assume they were bounded), $w$ was unbounded, and $p$ had the value $1$ with no fractional part.
    In particular, $x$ and $y$ are now collected in $X_{\mathit{bnd}}$.
    Since $x$ and $y$ are assumed to be bounded in $\breg$, all possible combinations of their integer values must now be considered. 
    For the purposes of this example, a single combination suffices.
    Let $\bar{h}$ (Line~\ref{line:fdp:foreachfunc}) be such that $\bar{h}(x) = 3$ and $\bar{h}(y) = 2$ (this allows to satisfy $\gamma$ on Line~\ref{line:fdp:removefromres}, since $h_{\breg}(w) = 5$, $h_{\breg}(p) = 1$, and $h_{\breg}(z) = h_{\breg}(s) = 4$).
    Since the order of the fractional parts of $x$ and $y$ in $\breg$ is unknown, all possible ordered partitions of $X_{\mathit{bnd}}$ and $X_{\mathit{oob}}$ must be explored, while preserving the order of the fractional parts of $z$ and $s$ (they are not reset in $t$).  
    One such ordered partition can be generated using \emph{part-regs} (Lines~\ref{line:fdp:allrinpermreg},\ref{line:fdp:resUres}).  
    Finally,
    $
    \breg = \{ q_0, h(x) = 3,$ $h(y)=  2, h(w) = 5, h(p) = 1, h(z) = h(s) = 4, X_{-1}, X_0, X_1, X_2, X_3 \}
    $, where 
    $X_{-1} = \{ w \}, 
    X_0 = \{ p \}, 
    X_1 = \{ z \}, 
    X_2 = \{ x, y \},$ and $
    X_3 = \{ s \}
    $.
\end{example}


\section{Proof of \autoref{thm:bounded_d_pred}}
\label{thm:proof_of_bounded_d_pred}
\begin{theorem}
   If the order in which clocks become unbounded is tracked, the number of immediate delay predecessors of a given region is at most three.
\end{theorem}

\begin{proof}
    We assume that all clocks in a given region have a value greater than zero, otherwise a delay predecessor cannot be computed.
    The proof is conducted by considering each class identified in \autoref{def:region_classes} and then analyzing all possible ways to transition backwards to predecessor regions using only delay transitions.

\begin{itemize}[noitemsep, topsep=0pt]
    \item \emph{Class} $\ca$:
    consider a region $\rega$ of class $\ca$.
    The immediate delay predecessor regions that can be reached backwards from $\rega$ belong to class $\cb$.
    Hence, let $\regb$ be a region of class $\cb$.
    Since $\ridx_{\rega} = 0$, it follows that from $\regb$ we can transition forward to $\rega{}$ only if $\ridx_{\regb} = 1$ holds (every clock must have a fractional part equal to zero when reaching $\rega$ from $\regb$).
    Since the only way to transition backwards from $\rega$ to $\regb$ is to remove all clocks from $X_{\rega, 0}$ and insert them into $X_{\regb, 1}$, while also decreasing the integer values of those clocks by 1, it follows that a region of class $\ca$ has exactly one immediate delay predecessor region of class $\cb$.

    \item \emph{Class} $\cb$:
    consider a region $\regb$ of class $\cb$.
    The immediate delay predecessor regions that can be reached backwards from $\regb$ belong to class $\ca$ or class $\cc$.
    Let $\rega$ be a region of class $\ca$ and $\regc$ be a region of class $\cc$.
\begin{enumerate}
    \item \label{thm1:classb1} \emph{From $\regb$ back to $\rega$}.
    This happens only if $\ridx_{\regb} = 1$.
    We must distinguish two further cases:
    (i) $\lidx_{\regb} = \lidx_{\rega}$: in this case, $X_{\rega, 0}$ contains the clocks that exited from the unit, thus entering in   $X_{\regb, 1}$. 
    Thus, to transition backwards from $\regb$ to $\rega$, it suffices to remove all clocks from $X_{\regb, 1}$ and insert them into $X_{\rega,0}$;
    (ii) $\lidx_{\regb} = \lidx_{\rega} + 1$: in this case, $X_{\rega, 0}$ contains the clocks that exited from the unit and the clocks that became unbounded, \emph{i.e.}, $X_{\rega, 0} = X_{\regb,-\lidx_{\regb}} \cup X_{\regb,1}$.
    Thus, to transition backwards from $\regb$ to $\rega$, it suffices to remove all clocks from $X_{\regb,-\lidx_{\regb}}$ and $X_{\regb, 1}$ and insert them into $X_{\rega,0}$.
    Hence, a region of class $\cb$ has at most two immediate delay predecessor regions of class $\ca{}$.

    \item \label{thm1:classb2} \emph{From $\regb$ back to $\regc$}.
    If $\ridx_{\regb} = 1$ holds, there is only one possible way to transition forward from $\regc$ to $\regb$, \emph{i.e.}, when $\ridx_{\regc} = 1$ and all clocks in $X_{\regc,0}$ must become unbounded.
    Hence, it is possible to transition backwards from $\regb$ to $\regc$ by removing all clocks from $X_{\regb,-\lidx_{\regb}}$ and inserting them into $X_{\regc,0}$.
    Otherwise, let us assume that $\ridx_{\regb} > 1$.
    We must distinguish three further cases:
    (i) $\lidx_{\regb} = \lidx_{\regc} \land \ridx_{\regb} = \ridx_{\regc} + 1$: in this case, $X_{\regc,0}$ contains exactly the clocks that will enter $X_{\regb,1}$; it is possible to transition backwards from $\regb$ to $\regc$ by removing all clocks from   $X_{\regb,1}$ and insert them into $X_{\regc,0}$;
    (ii) $\lidx_{\regb} = \lidx_{\regc} + 1 \land \ridx_{\regb} = \ridx_{\regc} + 1$: in this case, $X_{\regc,0}$ contains the clocks in $X_{\regb,-\lidx_{\regb}} \cup X_{\regb,1}$; it is possible to transition backwards from $\regb$ to $\regc$ by removing all clocks from $X_{\regb,-\lidx_{\regb}}$ and $X_{\regb,1}$ and insert them into $X_{\regc,0}$;
    (iii) $\lidx_{\regb} = \lidx_{\regc} + 1 \land \ridx_{\regb} = \ridx_{\regc}$: in this case, $X_{\regc,0}$ contains exactly the clocks that will become unbounded in $X_{\regb,-\lidx_{\regb}}$; it is possible to transition backwards from $\regb$ to $\regc$ by removing all clocks from $X_{\regb,-\lidx_{\regb}}$ and insert them into $X_{\regc,0}$.
    Hence, a region of class $\cb$ can either have one immediate delay predecessor of class $\cc$ to be added to the (up to two) immediate delay predecessors of class $\ca$ (when $\ridx_{\regb} = 1$), or at most three immediate delay predecessor regions of class $\cc$ when $\ridx_{\regb} > 1$.
\end{enumerate}
Combining the above results (\ref{thm1:classb1}) and (\ref{thm1:classb2}), it follows that a region of class $\cb$ has at most three immediate delay predecessors.

    \item \emph{Class} $\cc$:
    consider a region $\regc$ of class $\cc$.
    The immediate delay predecessor regions that can be reached backwards from $\regc$ belong to class $\cb$.
    Hence, let $\regb$ be a region of class $\cb$ such that $\ridx_{\regb} > 1$ (otherwise $\regb$ would reach a region of class $\ca$) and $\ridx_{\regb} = \ridx_{\regc} + 1$.
    The only way $\regc$ is reachable from $\regb$ is by letting the clocks in set $X_{\regb, \ridx_{\regb}}$ reach the next unit.
    Then, it is possible to transition backwards from $\regc$ to $\regb$ by removing all the clocks from $X_{\regc,0}$ and inserting them into $X_{\regb, \ridx_{\regb}}$, while also decreasing by 1 the integer value of those clocks.
    It follows that a region of class $\cc$ has exactly one immediate delay predecessor region of class $\cb$.

    \item \emph{Class} $\cu$:
    consider a region $\regu$ of class $\cu$.
    The immediate delay predecessor regions that can be reached backwards from $\regu$ belong to class $\ca$.
    Hence, let $\rega$ be a region of class $\ca$.
    The only way $\regu$ is reachable from $\rega$ is when all clocks in $X_{\rega,0}$ become unbounded simultaneously.
    It is possible to transition backwards from $\regu$ to $\rega$ by removing all clocks from $X_{\regu,-\lidx_{\regu}}$ and inserting them into $X_{\rega, 0}$.
    Hence, a region of class $\cu$ has exactly one immediate delay predecessor region of class $\ca$.
\end{itemize}

    It follows that regions in classes $\ca,\cc,\cu$ have one immediate delay predecessor; regions in class $\cb$ have at most three immediate delay predecessors.
    \hfill \qed
\end{proof}

\section{Proof of \autoref{thr:discrete_predecessor_complexity}}
\label{thm:proof_of_discrete_predecessor_complexity}
\begin{lemma}
\label{lem:discrete_complexity}
    Let $\aut$ be a \ac{ta} such that $n = \lvert \ackc \rvert$ and let $\reg$ be a region of $\aut$. If the order in which clocks become unbounded must be tracked, then the number of discrete predecessors of $\reg$ over a single transition is at most:
    $$
    \sum_{u = 0}^{n} \!
    \left(
    \sum_{i = 0}^{n - u} (\cmax + 1)^i (\cmax)^{n - u - i} 
    \binom{n}{n - u - i} \!\!
    \sum_{k = 0}^{n - u - i} k! \genfrac\{\}{0pt}{}{n - u - i}{k}
    \! \cdot \!
    \binom{n}{u} \!\!
    \sum_{w = 0}^{u} w! \genfrac\{\}{0pt}{}{u}{w} \!
    \right)
    $$
    where $\genfrac\{\}{0pt}{}{\alpha}{\beta} = \frac{1}{\beta!} \sum_{d = 0}^\beta (-1)^d \binom{\beta}{d} (\beta - d)^\alpha$ is a Stirling number of the second kind, that is, the number of ways to partition a set of $\alpha$ elements into $\beta$ non-empty subsets, for two given natural numbers $\alpha, \beta \in \mathbb{N}$.
\end{lemma}

\begin{proof}
We assume that all clocks of $\aut$ share the same maximum constant $\cmax$.
The lemma also applies when each clock has its own maximum constant.

The proof proceeds by enumerating all possible ordered partitions of bounded and unbounded clocks.
We consider the following transition over which discrete predecessors of a region $\reg$ are computed:
$(q, a, \gamma, Y, q_{\reg}) \in T$, where the guard is of the form $\gamma := x_1 \ge 0 \land \dots \land x_n \ge 0$, where $\ackc = \{ x_1, \dots, x_n \}$, and the set of reset clocks is $Y = \ackc$, \emph{i.e.}, all clocks are reset, allowing them to take on any possible value up to $\cmax$ and potentially become unbounded.

Let $u$ denote the total number of clocks that are unbounded in the predecessor of $\reg$ (for conciseness, in the following we simply say bounded or unbounded).

If $u = n$, meaning that all clocks are unbounded, the total number of ordered partitions is given by:
$\sum_{w = 0}^{u} w! \genfrac\{\}{0pt}{}{u}{w}$.
This corresponds to the definition of \emph{Fubini numbers}~\cite{Comtet2012,Fubini}, which count the number of ordered partitions of a given set.

If $u = 0$, meaning that all clocks are bounded, the total number of ordered partitions is given by:
$\sum_{i = 0}^{n} (\cmax + 1)^i (\cmax)^{n - i} \binom{n}{n - i} \sum_{k = 0}^{n - i} k! \genfrac\{\}{0pt}{}{n - i}{k}$.
The index $i$ represents the total number of clocks that have no fractional part.
Since these clocks do not contribute to the ordering, the number of ordered partitions only depends on the remaining $n - i$ clocks, hence the Fubini sum goes up to $n - i$.
To also account for all possible ways of choosing which $n - i$ clocks have a fractional part greater than zero, the sum must be repeated over all $\binom{n}{n - i}$ combinations.
Two additional terms $(\cmax + 1)^i$ and $(\cmax)^{n - i}$ account for all possible integer values that clocks can assume: clocks with no fractional part can take any value from $0$ to $\cmax$ (\emph{i.e.}, $\cmax + 1$ possibilities), while the remaining clocks, having a fractional part greater than zero, are limited to the range $[0, \cmax)$ (\emph{i.e.}, $\cmax$ possibilities).

In the general case, it holds that $0 \le u \le n$.
The final result is obtained by combining the contributions of both bounded and unbounded clocks for each possible value of $u$.
The contribution from the bounded clocks is given by:
$\sum_{i = 0}^{n - u} (\cmax + 1)^i (\cmax)^{n - u - i} 
\binom{n}{n - u - i}
\sum_{k = 0}^{n - u - i} k! \genfrac\{\}{0pt}{}{n - u - i}{k}$.
The contribution from the unbounded clocks is:
$
\binom{n}{u}
\sum_{w = 0}^{u} w! \genfrac\{\}{0pt}{}{u}{w}
$, where $\binom{n}{u}$ accounts for the choice of which $u$ clocks are unbounded.
Multiplying the bounded and unbounded contributions for each fixed $u$, and summing over all $u \in [0, n]$, yields the maximum number of discrete predecessors of $\reg$ over a single transition.
\hfill
\qed
\end{proof}

\begin{theorem}
    Let $\aut$ be a \ac{ta} such that $n = \lvert \ackc \rvert$ and let $\reg$ be a region of $\aut$. If the order in which clocks become unbounded must be tracked, then the number of discrete predecessors of $\reg$ over a single transition is $O\left( (\frac{\cmax + 1}{\ln 2}) ^{n}\cdot (n+1)!\right)$.  
\end{theorem}

\begin{proof}
Let $S_n$  be the formula of \autoref{lem:discrete_complexity}. 

The value $a_u= \sum_{w = 0}^{u} w! \genfrac\{\}{0pt}{}{u}{w}$ is the $u$-th Fubini number.
Therefore,  
$$
\sum_{k = 0}^{n - u - i} k! \genfrac\{\}{0pt}{}{n - u - i}{k} = a_{n-u-i}.
$$
Moreover, $(\cmax + 1)^i (\cmax)^{n - u - i}< (\cmax+1)^{n - u}$.
Hence, $S_n$ is less than:
$$
S'_n=\sum_{u = 0}^{n} (\cmax + 1)^{n-u} \binom{n}{u} a_u \sum_{i = 0}^{n - u}\binom{n}{n - u - i} a_{n-u-i}.
$$
Given that $a_n = \sum_{i=1}^n \binom{n}{i}a_{n-i}$, then   $\sum_{i = 0}^{n - u}\binom{n}{n - u - i} a_{n-u-i} = 2a_{n-u}$.
Therefore, $ S'_n$ is equal to:
$$
2\sum_{u = 0}^{n} (\cmax + 1)^{n-u} \binom{n}{u} a_{n-u} a_u \le 
2 \sum_{u = 0}^{n} (\cmax + 1)^{n-u} \cdot \sum_{u = 0}^{n} \binom{n}{u} a_{n-u} a_u
$$
$$
\le 2(\cmax + 1)^{n+1}\cdot \sum_{u = 0}^{n} \binom{n}{u} a_{n-u} a_u.
$$
We want to estimate
$
T_n = \sum_{u = 0}^{n} \binom{n}{u} a_{n-u} a_u,
$ 
when $n\to \infty$.
Fubini numbers satisfy the asymptotic estimate
$a_i = O\left( \frac{i!}{(\ln 2)^i} \right).$
Using this estimate, we write:
$$
\binom{n}{u} a_{n-u} a_u = 
O\left( \! \binom{n}{u} \cdot \frac{(n-u)!}{(\ln 2)^{n-u}} \cdot \frac{u!}{(\ln 2)^u} \! \right) \!
= O\left( \! \binom{n}{u} \cdot (n-u)! \cdot u! \cdot \frac{1}{(\ln 2)^n} \! \right)  \! .
$$
Since $\binom{n}{u} \cdot (n-u)! \cdot u! = n!$, each term in the sum $T_n$ is
$O\left( \frac{n!}{(\ln 2)^n} \right).$

There are $n+1$ such terms, hence $T_n$ is:
$$
O\left( (n+1) \cdot \frac{n!}{(\ln 2)^n} \right) = O\left( \frac{ (n+1)!}{(\ln 2)^n} \right)
$$
\emph{i.e.}, $S_n$ is $O\left( (\frac{ \cmax + 1}{\ln 2})^{n} \cdot (n+1)! \right).  
$
\hfill \qed
\end{proof}


\section{Complete Experimental Evaluation}
\label{app:complete_experimental_evaluation}
This section provides a detailed description of the manually constructed benchmarks and reports the complete version of the tables discussed in \autoref{sec:implementation_and_experimental_evaluation}.

Tables report \acl{vt} (VT), \acl{et} (ET), memory (Mem, representing the Maximum Resident Set Size), and the number of stored regions and states.
\Red{\ac{et} measures the time elapsed between issuing a verification command and receiving the result from a tool.}
Timings are in seconds, memory in MB.
The size of a network is denoted by $K$ (for \flower{}, $K$ indicates the number of clocks).
All values are averages over five runs, with a timeout of 600 s set on \ac{et}.
Timeouts are reported as TO, $\varepsilon$ denotes times below 0.001 s, and out-of-memory events are denoted by OOM.
Bold cells indicate the fastest \ac{vt} (on ties, the fastest \ac{et}).
KO indicates that an error occurred during verification.

\subsection{Punctual guards benchmarks}

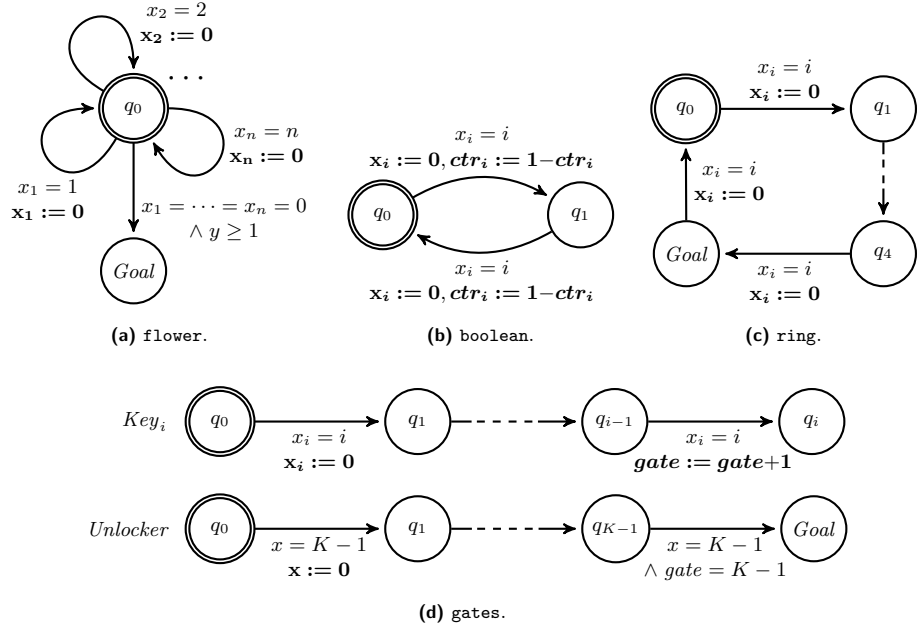
\begin{figure}[tb]
\centering
\begin{subfigure}{0.34\textwidth}
  \centering
  \begin{tikzpicture}[
    ->,
    >=stealth',
    shorten >=2pt, 
    auto,
    transform shape,
    align=center,
    scale=0.865,
    state/.style={thick, circle, draw, minimum size=1cm}
] 
\useasboundingbox (-1.9, -3.1) rectangle (2.7, 1.71);

  \node[state, accepting] (q) {$q_0$};
  \node[state, below of=q, yshift=-1.5cm] (goal) {$\mathit{Goal}$};
  \node[above=1cm of q, xshift=23pt, yshift=-37pt] (dots) {\Large $\mathbf{\cdots}$};
  \draw [thick] 
  (q) edge[loop, in=180, out=240, looseness=8]  node[below, xshift=0pt, yshift=-6pt]  {$x_1 = 1$ \\ $\resetClock{x_1}$} (q)
  
  (q) edge[loop, in=90, out=150, looseness=8] node[right, xshift=21pt, yshift=0pt] {$x_2 = 2$ \\ $\resetClock{x_2}$} (q)
  
  (q) edge[loop, in=300, out=0, looseness=8] node[right, xshift=1pt, yshift=5pt] {$x_n = n$ \\ $\resetClock{x_n}$} (q)
  
  (q) edge node[right, yshift=-13pt] {$x_1 = \cdots = x_n = 0$ \\ $\land \; y \geq 1$} (goal);

\end{tikzpicture}
  \caption{\flower{}.}
  \label{fig:flower_sub}
\end{subfigure}
\hfill
\begin{subfigure}{0.3\textwidth}
  \centering
  \begin{tikzpicture}[
    ->,
    >=stealth',
    shorten >=2pt, 
    auto,
    transform shape,
    align=center,
    scale=0.865,
    state/.style={thick, circle, draw, minimum size=1cm}
] 

  \node[state, accepting] (q) {$q_0$};
  \node[state, right=2cm of q] (q1) {$q_1$};
  
  \draw [thick] 
  (q) edge[in=150, out=30]  
  node[xshift=0pt, yshift=0pt]  
  {$x_i = i$ \\ $\mathbf{x}_{\boldsymbol{i}} \boldsymbol{:=} \mathbf{0}, \boldsymbol{ctr_i} \boldsymbol{:=} \mathbf{1} \mathbf{-} \boldsymbol{ctr_i}$} (q1)

  (q1) edge[in=330, out=210]  
  node[xshift=0pt, yshift=0pt]  
  {$x_i = i$ \\ $\mathbf{x}_{\boldsymbol{i}} \boldsymbol{:=} \mathbf{0}, \boldsymbol{ctr_i} \boldsymbol{:=} \mathbf{1} \mathbf{-} \boldsymbol{ctr_i}$} (q);

\end{tikzpicture}
  \caption{\boolea{}.}
  \label{fig:boolean_sub}
\end{subfigure}
\hfill
\begin{subfigure}{0.3\textwidth}
  \centering
  \begin{tikzpicture}[
    ->,
    >=stealth',
    shorten >=2pt, 
    auto,
    transform shape,
    align=center,
    scale=0.865,
    state/.style={thick, circle, draw, minimum size=1cm}
] 
 
  \node[state, accepting] (q) {$q_0$};
  \node[state, right=2cm of q] (q1) {$q_1$};
  \node[state, below=1.2cm of q1] (q2) {$q_4$};
  \node[state, left=2cm of q2] (q3) {$\mathit{Goal}$};
  
  \draw[thick, ->] (q) --  node[above] {$x_i = i$ \\ $\mathbf{x}_{\boldsymbol{i}} \boldsymbol{:=} \mathbf{0}$} (q1);
  \draw[thick, -] (q1) -- ($(q1)!0.35!(q2)$);
  \draw[thick, dashed, -] ($(q1)!0.35!(q2)$) -- ($(q1)!0.55!(q2)$);
  \draw[thick, ->] ($(q1)!0.55!(q2)$) -- (q2);
  \draw[thick, ->] (q2) --  node[below] {$x_i = i$ \\ $\mathbf{x}_{\boldsymbol{i}} \boldsymbol{:=} \mathbf{0}$} (q3);
  \draw[thick, ->] (q3) --  node[right] {$x_i = i$ \\ $\mathbf{x}_{\boldsymbol{i}} \boldsymbol{:=} \mathbf{0}$} (q);

\end{tikzpicture}
  \caption{\ring{}.}
  \label{fig:ring_sub}
\end{subfigure}

\vspace{1.5em}  

\begin{subfigure}{\textwidth}
  \centering
  \begin{tikzpicture}[
    ->,
    >=stealth',
    shorten >=2pt, 
    auto,
    transform shape,
    align=center,
    scale=0.865,
    state/.style={thick, circle, draw, minimum size=1cm}
] 
 
  \node[state, accepting] (q) {$q_0$};
  \node[state, left=0.1cm of q, draw=none] (k) {$\mathit{Key}_i$};
  \node[state, right=2cm of q] (q1) {$q_1$};
  \node[state, right=2cm of q1] (q2) {$q_{i-1}$};
  \node[state, right=2cm of q2] (q3) {$q_i$};
  
  \draw[thick, ->] (q) --  node[below] {$x_i = i$ \\ $\mathbf{x}_{\boldsymbol{i}} \boldsymbol{:=} \mathbf{0}$} (q1);
  \draw[thick, -] (q1) -- ($(q1)!0.3!(q2)$);
  \draw[thick, dashed, -] ($(q1)!0.3!(q2)$) -- ($(q1)!0.61!(q2)$);
  \draw[thick, ->] ($(q1)!0.61!(q2)$) -- (q2);
  \draw[thick, ->] (q2) --  node[below] {$x_i = i$ \\ $\boldsymbol{gate} \boldsymbol{:=} \boldsymbol{gate} \mathbf{+} \mathbf{1}$} (q3);
  
  \node[state, accepting, below=0.6cm of q] (p) {$q_0$};
  \node[state, left=0.1cm of p, draw=none] (u) {$\mathit{Unlocker}$};
  \node[state, right=2cm of p] (p1) {$q_1$};
  \node[state, right=2cm of p1] (p2) {$q_{K - 1}$};
  \node[state, right=2cm of p2] (p3) {$\mathit{Goal}$};
  
  \draw[thick, ->] (p) --  node[below] {$x = K - 1$ \\ $\resetClock{x}$} (p1);
  \draw[thick, -] (p1) -- ($(p1)!0.3!(p2)$);
  \draw[thick, dashed, -] ($(p1)!0.3!(p2)$) -- ($(p1)!0.61!(p2)$);
  \draw[thick, ->] ($(p1)!0.61!(p2)$) -- (p2);
  \draw[thick, ->] (p2) --  node[below] {$x = K - 1$ \\ $\land \; \mathit{gate} = K - 1$} (p3);

   
\end{tikzpicture}
  \caption{\gates{}.}
  \label{fig:gates_sub}
\end{subfigure}

\caption{Depiction of the benchmark \ac{ta} used in both \autoref{tab:punctual} and \autoref{tab:punctual_full}.
\boolea{}, \ring{}, and \gates{} are manually constructed networks of \ac{ta}, whereas \flower{} is a single \ac{ta} adapted from~\cite{J_rgensen_2012} (it is included here for completeness).}
\label{fig:overall}
\end{figure}

\autoref{fig:overall} depicts the manually created benchmarks used in \autoref{tab:punctual} and \autoref{tab:punctual_full}.
Those are manually constructed networks of \ac{ta}, except for \flower{}, which is a single \ac{ta} adapted from~\cite{J_rgensen_2012}.

\textbf{Flower}~The \flower{} \ac{ta} (\autoref{fig:flower_sub}) has two locations: $q_0$ (initial) and $\mathit{Goal}$, and it includes $n + 1$ clocks, denoted by $x_1, \dots, x_n, y$. 
Over $q_0$ there are $n$ self-loops, one for each clock $x_i$. 
In each self-loop, $x_i$ is compared to the constant $i$ (\emph{i.e.}, the guard is $x_i = i$) and reset along such transition.
To reach $\mathit{Goal}$, all clocks $x_i$ must simultaneously be exactly zero (with no fractional part), while $y$ must satisfy $y \ge 1$. 
Consequently, the integer value of $y$ when transitioning to $\mathit{Goal}$ is equal to the least common multiple of $1, 2, \dots, n$.

\textbf{Boolean}~The \boolea{} network consists of $K$ \acp{ta} of the form shown in \autoref{fig:boolean_sub}.  
For every index $i$, with $1 \le i \le K$, the $i$-th \ac{ta} flips the value of the integer variable $\mathit{ctr}_i$ between $0$ and $1$ whenever it transitions between location $q_0$ and location $q_1$.  
All integer variables are initialized to $0$ at the start of execution.
The query  
$
\exists \Diamond( \mathit{ctr}_1 = 1 \land \mathit{ctr}_2 = 1 \land \dots \land \mathit{ctr}_K = 1 )
$
is satisfied only if every variable $\mathit{ctr}_i$ is set to $1$.  
Since TChecker does not support integer variables in reachability queries, an equivalent property was formulated using locations:
$
\exists \Diamond( \mathit{Boolean}_1.q_1 \land \mathit{Boolean}_2.q_1 \land \dots \land \mathit{Boolean}_K.q_1 ).
$

\textbf{Ring}~The \ring{} network consists of $K$ \acp{ta} of the form shown in \autoref{fig:ring_sub}.
For brevity, the dashed edge abbreviates the additional locations and identical transitions that are not depicted (each \ac{ta} includes 6 locations and 6 transitions).

\textbf{Gates}~The \gates{} network consists of $K-1$ $\mathit{Key}_i$ \acp{ta}, for $1 \le i \le K-1$, and an $\mathit{Unlocker}$ \ac{ta}, as depicted in \autoref{fig:gates_sub}.  
For brevity, the dashed edges abbreviate additional locations and identical transitions that are not depicted.  
The total number of locations and transitions depends on the index $i$ for each $\mathit{Key}_i$ and on $K$ for the $\mathit{Unlocker}$.
Each $\mathit{Key}_i$ increments the integer variable $\mathit{gate}$ by $1$ during its last transition.  
The $\mathit{Unlocker}$ can reach the $\mathit{Goal}$ location only if $\mathit{gate}$ equals $K-1$, allowing its last transition to fire.

\subsection{Full benchmark tables}

\definecolor{lightgray}{gray}{0.95}
\setlength{\aboverulesep}{1pt}
\setlength{\belowrulesep}{1pt}
{\renewcommand{\arraystretch}{0.95}
{\setlength{\tabcolsep}{3.25pt}
\begin{table}[tb]
\centering
\caption{When all transitions have punctual guards, regions may outperform zones.
Here, \flower{} is a single \ac{ta}, while the others are networks of \ac{ta}.}
\label{tab:punctual_full}
\resizebox{\textwidth}{!}{
\begin{tabular}{%
    c|rrrr|rrrr|rrrr
}
\toprule
\multirow{2}{*}{\textbf{K}}
  & \multicolumn{4}{c|}{\textbf{TARZAN}}
  & \multicolumn{4}{c|}{\textbf{TChecker}}
  & \multicolumn{4}{c}{\textbf{UPPAAL}} \\
  & \multicolumn{1}{c}{\textbf{VT}}
  & \multicolumn{1}{c}{\textbf{ET}}
  & \multicolumn{1}{c}{\textbf{Mem}}
  & \multicolumn{1}{c|}{\textbf{Regions}}
  & \multicolumn{1}{c}{\textbf{VT}}
  & \multicolumn{1}{c}{\textbf{ET}}
  & \multicolumn{1}{c}{\textbf{Mem}}
  & \multicolumn{1}{c|}{\textbf{States}}
  & \multicolumn{1}{c}{\textbf{VT}}
  & \multicolumn{1}{c}{\textbf{ET}}
  & \multicolumn{1}{c}{\textbf{Mem}}
  & \multicolumn{1}{c}{\textbf{States}} \\
\midrule
\multicolumn{13}{c}{$\blacktriangleright$
{\normalsize \boolea{}}~\cite{tarzan2025} ::
[Query: $\exists \Diamond( \mathit{ctr}_1 = 1 \land \mathit{ctr}_2 = 1 \land \dots \land \mathit{ctr}_K = 1 )$]
$\blacktriangleleft$} \\
\midrule
\rowcolor{white}
2  & $\varepsilon$ & 0.013 & 7.62 & 9
   & ${\boldsymbol \varepsilon}$ & \textbf{0.005} & \textbf{5.92} & \textbf{6}
   & $\varepsilon$ & 0.051 & 18.80 & 6 \\
\rowcolor{lightgray}
4  & $\varepsilon$ & 0.013 & 7.78 & 24
   & ${\boldsymbol \varepsilon}$ & \textbf{0.007} & \textbf{6.47} & \textbf{185}
   & 0.001 & 0.052 & 18.85 & 185 \\
\rowcolor{white}
6  & \textbf{0.003} & \textbf{0.016} & \textbf{13.42} & \textbf{425}
   & 0.088 & 0.096 & 13.80 & 8502
   & 0.075 & 0.128 & 19.91 & 8502 \\
\rowcolor{lightgray}
8  & \textbf{0.013} & \textbf{0.027} & \textbf{29.10} & \textbf{1009}
   & 75.461 & 75.980 & 220.82 & 566763
   & 77.204 & 77.523 & 81.79 & 566763 \\
\rowcolor{white}
10  & \textbf{0.070} & \textbf{0.086} & \textbf{110.29} & \textbf{4050}
   & --- & TO & --- & ---
   & --- & TO & --- & --- \\
\rowcolor{lightgray}
12  & \textbf{0.341} & \textbf{0.361} & \textbf{476.73} & \textbf{15961}
   & --- & TO & --- & ---
   & --- & TO & --- & --- \\
\rowcolor{white}
14  & \textbf{0.737} & \textbf{0.764} & \textbf{965.29} & \textbf{28679}
   & --- & TO & --- & ---
   & --- & TO & --- & --- \\
\rowcolor{lightgray}
16  & \textbf{1.578} & \textbf{1.616} & \textbf{1788.46} & \textbf{43919}
   & --- & TO & --- & ---
   & --- & TO & --- & --- \\
\midrule
\multicolumn{13}{c}{$\blacktriangleright$
{\normalsize \flower{}}~\cite{J_rgensen_2012} :: 
[Query: $\exists \Diamond ( \mathit{Flower.Goal} )$]
$\blacktriangleleft$} \\
\midrule
\rowcolor{white}
3  & $\varepsilon$ & 0.013 & 7.64 & 13
   & ${\boldsymbol \varepsilon}$ & \textbf{0.005} & \textbf{5.84} & \textbf{7}
   & $\varepsilon$ & 0.055 & 18.80 & 7 \\
\rowcolor{lightgray}
5  & ${\boldsymbol \varepsilon}$ & \textbf{0.013} & \textbf{7.92} & \textbf{95}
   & 0.001 & 0.007 & 6.35 & 150
   & 0.001 & 0.056 & 18.82 & 150 \\
\rowcolor{white}
7  & ${\boldsymbol \varepsilon}$ & \textbf{0.022} & \textbf{10.65} & \textbf{573}
   & 0.322 & 0.330 & 13.12 & 3508
   & 0.399 & 0.455 & 19.76 & 3508 \\
\rowcolor{lightgray}
9  & \textbf{0.007} & \textbf{0.022} & \textbf{58.62} & \textbf{9161}
   & --- & TO & --- & ---
   & --- & TO & --- & --- \\
\rowcolor{white}
11  & \textbf{0.028} & \textbf{0.044} & \textbf{177.34} & \textbf{30761}
   & --- & TO & --- & ---
   & --- & TO & --- & --- \\
\rowcolor{lightgray}
13  & \textbf{0.412} & \textbf{0.450} & \textbf{2034.18} & \textbf{370331}
   & --- & TO & --- & ---
   & --- & TO & --- & --- \\
\rowcolor{white}
15  & \textbf{7.188} & \textbf{7.398} & \textbf{12900.36} & \textbf{5142671}
   & --- & TO & --- & ---
   & --- & TO & --- & --- \\
\rowcolor{lightgray}
17  & \textbf{16.212} & \textbf{16.562} & \textbf{12942.48} & \textbf{11023829}
   & --- & TO & --- & ---
   & --- & TO & --- & --- \\
\midrule
\multicolumn{13}{c}{$\blacktriangleright$
{\normalsize \gates{}}~\cite{tarzan2025} ::
[Query: $\exists \Diamond ( \mathit{Unlocker.Goal} )$]
$\blacktriangleleft$} \\
\midrule
\rowcolor{white}
3  & $\varepsilon$ & 0.013 & 7.79 & 26
   & ${\boldsymbol \varepsilon}$ & \textbf{0.006} & \textbf{5.99} & \textbf{18}
   & $\varepsilon$ & 0.054 & 18.86 & 18 \\
\rowcolor{lightgray}
5  & ${\boldsymbol \varepsilon}$ & \textbf{0.013} & \textbf{8.84} & \textbf{76}
   & 0.001 & 0.008 & 7.19 & 549
   & 0.001 & 0.055 & 19.01 & 549 \\
\rowcolor{white}
7  & ${\boldsymbol \varepsilon}$ & \textbf{0.014} & \textbf{11.23} & \textbf{151}
   & 0.110 & 0.122 & 21.08 & 32979
   & 0.035 & 0.092 & 22.79 & 32979 \\
\rowcolor{lightgray}
9  & \textbf{0.001} & \textbf{0.015} & \textbf{15.43} & \textbf{251}
   & 47.895 & 48.771 & 806.72 & 3071646
   & 5.228 & 6.509 & 451.71 & 3071646 \\
\rowcolor{white}
11  & \textbf{0.002} & \textbf{0.016} & \textbf{21.87} & \textbf{379}
   & --- & TO & --- & ---
   & --- & TO & --- & --- \\
\rowcolor{lightgray}
13  & \textbf{0.004} & \textbf{0.018} & \textbf{31.31} & \textbf{537}
   & --- & TO & --- & ---
   & --- & TO & --- & --- \\
\rowcolor{white}
15  & \textbf{0.006} & \textbf{0.020} & \textbf{44.01} & \textbf{708}
   & --- & TO & --- & ---
   & --- & TO & --- & --- \\
\rowcolor{lightgray}
17  & \textbf{0.009} & \textbf{0.023} & \textbf{60.85} & \textbf{922}
   & --- & TO & --- & ---
   & --- & TO & --- & --- \\
\midrule
\multicolumn{13}{c}{$\blacktriangleright$
{\normalsize \ring{}}~\cite{tarzan2025} ::
[Query: $\exists \Diamond ( P_0.\mathit{Goal} \land P_1.\mathit{Goal} \land ... \land P_K.\mathit{Goal} )$]
$\blacktriangleleft$} \\
\midrule
\rowcolor{white}
2  & $\varepsilon$ & 0.013 & 7.84 & 59
   & ${\boldsymbol \varepsilon}$ &\textbf{0.006} & \textbf{6.14} & \textbf{45}
   & $\varepsilon$ & 0.054 & 18.83 & 45 \\
\rowcolor{lightgray}
4  & \textbf{0.001} & \textbf{0.015} & \textbf{12.40} & \textbf{542}
   & 0.020 & 0.029 & 11.56 & 7738
   & 0.010 & 0.065 & 19.52 & 7738 \\
\rowcolor{white}
6  & \textbf{0.011} & \textbf{0.025} & \textbf{48.96} & \textbf{3397}
   & 12.137 & 12.561 & 268.93 & 1180575
   & 3.406 & 3.715 & 128.87 & 1180575 \\
\rowcolor{lightgray}
8  & \textbf{0.204} & \textbf{0.225} & \textbf{687.77} & \textbf{47063}
   & --- & TO & --- & ---
   & --- & TO & --- & --- \\
\rowcolor{white}
10  & \textbf{0.272} & \textbf{0.295} & \textbf{897.00} & \textbf{52669}
   & --- & TO & --- & ---
   & --- & TO & --- & --- \\
\rowcolor{lightgray}
12  & \textbf{6.901} & \textbf{7.097} & \textbf{13049.36} & \textbf{1067623}
   & --- & TO & --- & ---
   & --- & TO & --- & --- \\
\rowcolor{white}
14  & --- & --- & OOM & ---
   & --- & TO & --- & ---
   & --- & TO & --- & --- \\
\rowcolor{lightgray}
16  & --- & --- & OOM & ---
   & --- & TO & --- & ---
   & --- & TO & --- & --- \\
\bottomrule
\end{tabular}
}
\end{table}
}
}

\definecolor{lightgray}{gray}{0.95}
\setlength{\aboverulesep}{1pt}
\setlength{\belowrulesep}{1pt}
{\renewcommand{\arraystretch}{0.95}
{\setlength{\tabcolsep}{3.25pt}
\begin{table}[tb]
\centering
\caption{The results are derived from networks of closed \acp{ta}, where guards may either be punctual or contain non-strict inequalities.
In this setting, \toolname{} can handle both large networks (\medical{}) and large constants (\mpeg{}).}
\label{tab:closed_ta_full}
\resizebox{\textwidth}{!}{
\begin{tabular}{%
    c|rrrr|rrrr|rrrr
}
\toprule
\multirow{2}{*}{\textbf{K}}
  & \multicolumn{4}{c|}{\textbf{TARZAN}}
  & \multicolumn{4}{c|}{\textbf{TChecker}}
  & \multicolumn{4}{c}{\textbf{UPPAAL}} \\
  & \multicolumn{1}{c}{\textbf{VT}}
  & \multicolumn{1}{c}{\textbf{ET}}
  & \multicolumn{1}{c}{\textbf{Mem}}
  & \multicolumn{1}{c|}{\textbf{Regions}}
  & \multicolumn{1}{c}{\textbf{VT}}
  & \multicolumn{1}{c}{\textbf{ET}}
  & \multicolumn{1}{c}{\textbf{Mem}}
  & \multicolumn{1}{c|}{\textbf{States}}
  & \multicolumn{1}{c}{\textbf{VT}}
  & \multicolumn{1}{c}{\textbf{ET}}
  & \multicolumn{1}{c}{\textbf{Mem}}
  & \multicolumn{1}{c}{\textbf{States}} \\
\midrule
\multicolumn{13}{c}{$\blacktriangleright$
{\normalsize \medical{}}~\cite{tapaal2025} ::
[Query: $\exists \Diamond ( \mathit{Patient}_1.\mathit{Done} \land \mathit{Patient}_2.\mathit{Done} \land \dots \land \mathit{Patient}_{K - 2}.\mathit{Done} ) $]
$\blacktriangleleft$} \\
\midrule
\rowcolor{white}
12  & 0.007 & 0.021 & 33.39 & 660
   & 0.011 & 0.030 & 22.63 & 639
   & \textbf{0.002} & \textbf{0.059} & \textbf{20.83} & \textbf{582} \\
\rowcolor{lightgray}
22  & 0.056 & 0.072 & 204.97 & 2715
   & 0.316 & 0.560 & 173.75 & 3774
   & \textbf{0.032} & \textbf{0.093} & \textbf{31.76} & \textbf{2862} \\
\rowcolor{white}
32  & \textbf{0.218} & \textbf{0.241} & \textbf{802.55} & \textbf{7170}
   & 2.693 & 3.591 & 947.61 & 11409
   & 0.232 & 0.305 & 79.99 & 7842 \\
\rowcolor{lightgray}
42  & \textbf{0.609} & \textbf{0.647} & \textbf{2271.65} & \textbf{15025}
   & 13.571 & 16.096 & 2517.18 & 25544
   & 1.275 & 1.379 & 230.07 & 16522 \\
\rowcolor{white}
52  & \textbf{1.432} & \textbf{1.495} & \textbf{5210.36} &\textbf{27280}
   & 48.690 & 55.370 & 5724.98 & 48179
   & 5.524 & 5.695 & 598.03 & 29902 \\
\rowcolor{lightgray}
62  & \textbf{2.860} & \textbf{2.983} & \textbf{10440.43} & \textbf{44935}
   & 141.818 & 156.903 & 10778.80 & 81314
   & 19.325 & 19.648 & 1366.31 & 48982 \\
\rowcolor{white}
72  & \textbf{5.781} & \textbf{5.959} & \textbf{12891.17} & \textbf{68990}
   & 355.827 & 389.223 & 16109.45 & 126949
   & 54.200 & 54.838 & 2799.57 & 74762 \\
\rowcolor{lightgray}
82  & \textbf{10.622} & \textbf{10.889} & \textbf{13519.50} & \textbf{100445}
   & --- & TO & --- & ---
   & 133.114 & 134.367 & 5257.41 & 108242 \\
\rowcolor{white}
92  & \textbf{17.877} & \textbf{18.264} & \textbf{13811.47} & \textbf{140300}
   & KO & KO & KO & KO
   & 293.285 & 295.567 & 9219.25 & 150422 \\
\rowcolor{lightgray}
102  & \textbf{27.918} & \textbf{28.466} & \textbf{13780.57} & \textbf{189555}
   & KO & KO & KO & KO
   & 593.317 & 597.306 & 12364.67 & 202302 \\
\midrule
\multicolumn{13}{c}{$\blacktriangleright$
{\normalsize \mpeg{}}~\cite{tapaal2025} ::
[Query: $\exists \Diamond ( \mathit{BFrame}_1.\mathit{Bout} \land \mathit{BFrame}_2.\mathit{Bout} \land \dots \land \mathit{BFrame}_{K - 4}.\mathit{Bout} )$]
$\blacktriangleleft$} \\
\midrule
\rowcolor{white}
8  & 0.365 & 0.385 & 451.46 & 26916
   & 0.003 & 0.014 & 8.96 & 531
   & \textbf{0.001} & \textbf{0.058} & \textbf{19.46} & \textbf{537} \\
\rowcolor{lightgray}
12  & \textbf{0.796} & \textbf{0.819} & \textbf{718.63} & \textbf{27238}
   & 8.405 & 8.663 & 252.70 & 383103
   & 1.621 & 1.766 & 112.33 & 383109 \\
\rowcolor{white}
16  & \textbf{1.364} & \textbf{1.392} & \textbf{1032.46} & \textbf{27704}
   & --- & TO & --- & ---
   & --- & TO & --- & --- \\
\rowcolor{lightgray}
20  & \textbf{2.521} & \textbf{2.558} & \textbf{1575.82} & \textbf{28314}
   & --- & TO & --- & ---
   & --- & TO & --- & --- \\
\rowcolor{white}
24  & \textbf{3.865} & \textbf{3.910} & \textbf{2085.50} & \textbf{29068}
   & --- & TO & --- & ---
   & --- & TO & --- & --- \\
\bottomrule
\end{tabular}
}
\end{table}
}
}

\definecolor{lightgray}{gray}{0.95}
\setlength{\aboverulesep}{1pt}
\setlength{\belowrulesep}{1pt}
{\renewcommand{\arraystretch}{0.95}
{\setlength{\tabcolsep}{3.25pt}
\begin{table}[tb]
\centering
\caption{Large constants negatively impact the \ac{vt} of regions, unlike zones.
Here, the first row of \pagerank{} corresponds to the variant with truncated constants.}
\label{tab:big_constants_full}
\resizebox{\textwidth}{!}{
\begin{tabular}{%
    c|rrrr|rrrr|rrrr
}
\toprule
\multirow{2}{*}{\textbf{K}}
  & \multicolumn{4}{c|}{\textbf{TARZAN}}
  & \multicolumn{4}{c|}{\textbf{TChecker}}
  & \multicolumn{4}{c}{\textbf{UPPAAL}} \\
  & \multicolumn{1}{c}{\textbf{VT}}
  & \multicolumn{1}{c}{\textbf{ET}}
  & \multicolumn{1}{c}{\textbf{Mem}}
  & \multicolumn{1}{c|}{\textbf{Regions}}
  & \multicolumn{1}{c}{\textbf{VT}}
  & \multicolumn{1}{c}{\textbf{ET}}
  & \multicolumn{1}{c}{\textbf{Mem}}
  & \multicolumn{1}{c|}{\textbf{States}}
  & \multicolumn{1}{c}{\textbf{VT}}
  & \multicolumn{1}{c}{\textbf{ET}}
  & \multicolumn{1}{c}{\textbf{Mem}}
  & \multicolumn{1}{c}{\textbf{States}} \\
\midrule
\multicolumn{13}{c}{$\blacktriangleright$
{\normalsize \csma{}}~\cite{kiviriga2021randomized} ::
[Query: $\exists \Diamond \left( \parbox{7.5cm}{\begin{math}P_1.\mathit{senderRetry} \land P_2.\mathit{senderRetry} \land P_3.\mathit{senderTransm} \end{math} \\ \begin{math}\land \; P_4.\mathit{senderRetry} \land \dots \land P_7.\mathit{senderRetry}\end{math}} \right)$]
$\blacktriangleleft$} \\
\midrule
\rowcolor{white}
21  & 0.198 & 0.215 & 364.96 & 20936
   & 0.004 & 0.038 & 30.32 & 232
   & \textbf{0.001} & \textbf{0.061} & \textbf{20.24} & \textbf{269} \\
\rowcolor{lightgray}
23  & 0.262 & 0.281 & 482.81 & 25422
   & 0.005 & 0.053 & 40.75 & 254
   & \textbf{0.001} & \textbf{0.061} & \textbf{20.10} & \textbf{318} \\
\rowcolor{white}
26  & 0.388 & 0.409 & 700.02 & 32931
   & 0.008 & 0.079 & 62.17 & 287
   & \textbf{0.002} & \textbf{0.062} & \textbf{20.64} & \textbf{399} \\
\rowcolor{lightgray}
31  & 0.674 & 0.701 & 1194.45 & 47526
   & 0.018 & 0.218 & 117.03 & 371
   & \textbf{0.002} & \textbf{0.064} & \textbf{21.11} & \textbf{554} \\
\rowcolor{white}
51  & 3.364 & 3.431 & 5351.35 & 131906
   & 0.293 & 1.386 & 837.12 & 1381
   & \textbf{0.010} & \textbf{0.077} & \textbf{23.89} & \textbf{1424} \\
\midrule
\multicolumn{13}{c}{$\blacktriangleright$
{\normalsize \pagerank{}}~\cite{Baresi2020} ::
[Query: $\exists \Diamond ( \mathit{Stage}_7.\mathit{Completed} )$]
$\blacktriangleleft$} \\
\midrule
\rowcolor{white}
9  & 0.068 & 0.085 & 179.33 & 4193
   & 0.002 & 0.086 & 58.54 & 1129
   & \textbf{0.002} & \textbf{0.076} & \textbf{22.89} & \textbf{1129} \\
\rowcolor{lightgray}
9  & 4.227 & 4.361 & 10950.36 & 257311
   & 0.002 & 0.085 & 58.53 & 1129
   & \textbf{0.002} & \textbf{0.076} & \textbf{23.01} & \textbf{1129} \\
\bottomrule
\end{tabular}
}
\end{table}
}
}

\definecolor{lightgray}{gray}{0.95}
\setlength{\aboverulesep}{1pt}
\setlength{\belowrulesep}{1pt}
{\renewcommand{\arraystretch}{0.95}
{\setlength{\tabcolsep}{3.25pt}
\begin{table}[tb]
\centering
\caption{Safety properties require \toolname{} to visit the entire reachable state space.}
\label{tab:full_state_space_full}
\resizebox{\textwidth}{!}{
\begin{tabular}{%
    c|rrrr|rrrr|rrrr
}
\toprule
\multirow{2}{*}{\textbf{K}}
  & \multicolumn{4}{c|}{\textbf{TARZAN}}
  & \multicolumn{4}{c|}{\textbf{TChecker}}
  & \multicolumn{4}{c}{\textbf{UPPAAL}} \\
  & \multicolumn{1}{c}{\textbf{VT}}
  & \multicolumn{1}{c}{\textbf{ET}}
  & \multicolumn{1}{c}{\textbf{Mem}}
  & \multicolumn{1}{c|}{\textbf{Regions}}
  & \multicolumn{1}{c}{\textbf{VT}}
  & \multicolumn{1}{c}{\textbf{ET}}
  & \multicolumn{1}{c}{\textbf{Mem}}
  & \multicolumn{1}{c|}{\textbf{States}}
  & \multicolumn{1}{c}{\textbf{VT}}
  & \multicolumn{1}{c}{\textbf{ET}}
  & \multicolumn{1}{c}{\textbf{Mem}}
  & \multicolumn{1}{c}{\textbf{States}} \\
\midrule
\multicolumn{13}{c}{$\blacktriangleright$
{\normalsize \fischer{}}~\cite{uppaalSite} ::
[Query: $\exists \Diamond ( \mathit{Fischer}_1.\mathit{cs} \land \mathit{Fischer}_2.\mathit{cs} ) $]
$\blacktriangleleft$} \\
\midrule
\rowcolor{white}
2  & 0.001 & 0.015 & 9.60 & 542
   & ${\boldsymbol \varepsilon}$ & \textbf{0.006} & \textbf{6.01} & \textbf{18}
   & 0.001 & 0.055 & 18.87 & 18 \\
\rowcolor{lightgray}
3  & 0.025 & 0.039 & 55.62 & 9324
   & ${\boldsymbol \varepsilon}$ & \textbf{0.006} & \textbf{6.26} & \textbf{65}
   & 0.001 & 0.055 & 18.91 & 65 \\
\rowcolor{white}
4  & 0.644 & 0.669 & 1098.36 & 180032
   & ${\boldsymbol \varepsilon}$ & \textbf{0.007} & \textbf{6.75} & \textbf{220}
   & 0.001 & 0.055 & 18.96 & 220 \\
\rowcolor{lightgray}
5  & 19.627 & 19.849 & 12948.99 & 3920652
   & 0.004 & 0.011 & 9.30 & 727
   & \textbf{0.002} & \textbf{0.057} & \textbf{19.07} & \textbf{727} \\
\rowcolor{white}
6  & --- & --- & OOM & ---
   & 0.024 & 0.032 & 12.25 & 2378
   & \textbf{0.008} & \textbf{0.064} & \textbf{19.43} & \textbf{2378} \\
\midrule
\multicolumn{13}{c}{$\blacktriangleright$
{\normalsize \lynch{}}~\cite{farkas2018towards} ::
[Query: $\exists \Diamond ( P_1.\mathit{CS}7 \land P_2.\mathit{CS}7 ) $]
$\blacktriangleleft$} \\
\midrule
\rowcolor{white}
2  & 0.076 & 0.091 & 152.12 & 32268
   & ${\boldsymbol \varepsilon}$ & \textbf{0.007} & \textbf{6.19} & \textbf{38}
   & $\varepsilon$ & 0.058 & 18.95 & 38 \\
\rowcolor{lightgray}
3  & 20.137 & 20.373 & 12857.20 & 5289449
   & ${\boldsymbol \varepsilon}$ & \textbf{0.008} & \textbf{6.49} & \textbf{125}
   & 0.001 & 0.056 & 19.05 & 125 \\
\rowcolor{white}
4  & --- & --- & OOM & ---
   & \textbf{0.001} & \textbf{0.010} & \textbf{7.38} & \textbf{380}
   & 0.001 & 0.057 & 19.15 & 380 \\
\bottomrule
\end{tabular}
}
\end{table}
}
}

\definecolor{lightgray}{gray}{0.95}
\setlength{\aboverulesep}{1pt}
\setlength{\belowrulesep}{1pt}
{\renewcommand{\arraystretch}{0.95}
{\setlength{\tabcolsep}{3.25pt}
\begin{table}[tb]
\centering
\caption{Symmetry reduction renders some benchmarks tractable with regions. For each benchmark, symmetric \ac{ta} are identical. The maximum constant for each $\mathit{Viking}_i$ in \bridge{} is 20. The \train{} benchmark was run in \ac{bfs}.}
\label{tab:symmetry}
\resizebox{\textwidth}{!}{
\begin{tabular}{%
    c|rrrr|rrrr|rrrr
}
\toprule
\multirow{2}{*}{\textbf{K}}
  & \multicolumn{4}{c|}{\textbf{TARZAN}}
  & \multicolumn{4}{c|}{\textbf{TChecker}}
  & \multicolumn{4}{c}{\textbf{UPPAAL}} \\
  & \multicolumn{1}{c}{\textbf{VT}}
  & \multicolumn{1}{c}{\textbf{ET}}
  & \multicolumn{1}{c}{\textbf{Mem}}
  & \multicolumn{1}{c|}{\textbf{Regions}}
  & \multicolumn{1}{c}{\textbf{VT}}
  & \multicolumn{1}{c}{\textbf{ET}}
  & \multicolumn{1}{c}{\textbf{Mem}}
  & \multicolumn{1}{c|}{\textbf{States}}
  & \multicolumn{1}{c}{\textbf{VT}}
  & \multicolumn{1}{c}{\textbf{ET}}
  & \multicolumn{1}{c}{\textbf{Mem}}
  & \multicolumn{1}{c}{\textbf{States}} \\
\midrule
\multicolumn{13}{c}{$\blacktriangleright$
{\normalsize \train{}}~\cite{farkas2018towards} ::
[Query: $\exists \Diamond ( \mathit{Controller.controller3} ) $]
$\blacktriangleleft$} \\
\midrule
\rowcolor{white}
6  & 0.028 & 0.042 & 57.87 & 3803
   & 0.003 & 0.011 & 7.89 & 970
   & ${\boldsymbol \varepsilon}$ & \textbf{0.055} & \textbf{19.09} & \textbf{29} \\
\rowcolor{lightgray}
10  & 0.227 & 0.249 & 261.72 & 18672
   & 1.993 & 2.021 & 91.89 & 110096
   & \textbf{0.001} & \textbf{0.056} & \textbf{19.17} & \textbf{30} \\
\rowcolor{white}
14  & 0.553 & 0.572 & 370.96 & 33907
   & 159.005 & 159.662 & 2347.05 & 2087686
   & \textbf{0.001} & \textbf{0.059} & \textbf{19.26} & \textbf{30} \\
\rowcolor{lightgray}
18  & 0.977 & 0.997 & 464.34 & 49143
   & --- & TO & --- & ---
   & \textbf{0.001} & \textbf{0.057} & \textbf{19.30} & \textbf{30} \\
\rowcolor{white}
22  & 1.520 & 1.542 & 559.03 & 64379
   & --- & TO & --- & ---
   & \textbf{0.002} & \textbf{0.057} & \textbf{19.40} & \textbf{30} \\
\midrule
\multicolumn{13}{c}{$\blacktriangleright$
{\normalsize \bridge{}}~\cite{uppaalSite} ::
[Query: $\exists \Diamond ( \mathit{Viking}_1.\mathit{safe} \land \mathit{Viking}_2.\mathit{safe} \land \dots \land \mathit{Viking}_{K - 1}.\mathit{safe} ) $]
$\blacktriangleleft$} \\
\midrule
\rowcolor{white}
5  & 0.004 & 0.017 & 19.00 & 852
   & ${\boldsymbol \varepsilon}$ & \textbf{0.006} & \textbf{6.24} & \textbf{41}
   & $\varepsilon$ & 0.055 & 19.03 & 23 \\
\rowcolor{lightgray}
9  & 0.018 & 0.033 & 68.88 & 2271
   & ${\boldsymbol \varepsilon}$ & \textbf{0.008} & \textbf{6.44} & \textbf{149}
   & 0.001 & 0.057 & 19.07 & 55 \\
\rowcolor{white}
13  & 0.045 & 0.060 & 158.88 & 3883
   & ${\boldsymbol \varepsilon}$ & \textbf{0.011} & \textbf{6.97} & \textbf{321}
   & 0.001 & 0.056 & 19.13 & 87 \\
\rowcolor{lightgray}
17  & 0.086 & 0.104 & 286.59 & 5687
   & \textbf{0.002} & \textbf{0.017} & \textbf{17.07} & \textbf{557}
   & 0.002 & 0.058 & 19.17 & 119 \\
\rowcolor{white}
21  & 0.144 & 0.163 & 448.89 & 7683
   & 0.004 & 0.030 & 31.56 & 857
   & \textbf{0.003} & \textbf{0.060} & \textbf{19.26} & \textbf{151} \\
\bottomrule
\end{tabular}
}
\end{table}
}
}

\definecolor{lightgray}{gray}{0.95}
\setlength{\aboverulesep}{1pt}
\setlength{\belowrulesep}{1pt}
{\renewcommand{\arraystretch}{0.95}
{\setlength{\tabcolsep}{3.25pt}
\begin{table}[tb]
\centering
\caption{In small to medium-sized benchmarks, regions perform comparably with zones.
Here, \exSITH{} and \simple{} are single \ac{ta}.}
\label{tab:others}
\resizebox{\textwidth}{!}{
\begin{tabular}{%
    c|rrrr|rrrr|rrrr
}
\toprule
\multirow{2}{*}{\textbf{K}}
  & \multicolumn{4}{c|}{\textbf{TARZAN}}
  & \multicolumn{4}{c|}{\textbf{TChecker}}
  & \multicolumn{4}{c}{\textbf{UPPAAL}} \\
  & \multicolumn{1}{c}{\textbf{VT}}
  & \multicolumn{1}{c}{\textbf{ET}}
  & \multicolumn{1}{c}{\textbf{Mem}}
  & \multicolumn{1}{c|}{\textbf{Regions}}
  & \multicolumn{1}{c}{\textbf{VT}}
  & \multicolumn{1}{c}{\textbf{ET}}
  & \multicolumn{1}{c}{\textbf{Mem}}
  & \multicolumn{1}{c|}{\textbf{States}}
  & \multicolumn{1}{c}{\textbf{VT}}
  & \multicolumn{1}{c}{\textbf{ET}}
  & \multicolumn{1}{c}{\textbf{Mem}}
  & \multicolumn{1}{c}{\textbf{States}} \\
\midrule
\multicolumn{13}{c}{$\blacktriangleright$
{\normalsize \andor{}}~\cite{farkas2018towards} ::
[Query: $\exists \Diamond ( \mathit{false} ) $]
$\blacktriangleleft$} \\
\midrule
\rowcolor{white}
3  & 0.001 & 0.016 & 14.50 & 640
   & ${\boldsymbol \varepsilon}$ & \textbf{0.007} & \textbf{6.35} & \textbf{9}
   & $\varepsilon$ & 0.052 & 18.96 & 8 \\
\midrule
\multicolumn{13}{c}{$\blacktriangleright$
{\normalsize \exSITH{}}~\cite{farkas2018towards} ::
[Query: $\exists \Diamond ( A.\mathit{qBad} ) $]
$\blacktriangleleft$} \\
\midrule
\rowcolor{white}
1  & $\varepsilon$ & 0.013 & 8.10 & 225
   & ${\boldsymbol \varepsilon}$ & \textbf{0.006} & \textbf{5.87} & \textbf{4}
   & $\varepsilon$ & 0.054 & 18.79 & 4 \\
\midrule
\multicolumn{13}{c}{$\blacktriangleright$
{\normalsize \latch{}}~\cite{farkas2018towards} ::
[Query: $\exists \Diamond ( \mathit{Latch}_1.\mathit{LatchD1E1} ) $]
$\blacktriangleleft$} \\
\midrule
\rowcolor{white}
7  & 0.009 & 0.024 & 59.34 & 2025
   & ${\boldsymbol \varepsilon}$ & \textbf{0.008} & \textbf{6.36} & \textbf{7}
   & $\varepsilon$ & 0.056 & 19.03 & 7 \\
\midrule
\multicolumn{13}{c}{$\blacktriangleright$
{\normalsize \maler{}}~\cite{farkas2018towards} ::
[Query: $\exists \Diamond ( \mathit{Job}_1.\mathit{End}1 \land \mathit{Job}_2.\mathit{End}2 \land \mathit{Job}_3.\mathit{End}3 \land \mathit{Job}_4.\mathit{End}4 ) $]
$\blacktriangleleft$} \\
\midrule
\rowcolor{white}
4  & $\varepsilon$ & 0.014 & 8.38 & 114
   & ${\boldsymbol \varepsilon}$ & \textbf{0.008} & \textbf{6.35} & \textbf{75}
   & $\varepsilon$ & 0.056 & 19.11 & 75 \\
\midrule
\multicolumn{13}{c}{$\blacktriangleright$
{\normalsize \rcp{}}~\cite{farkas2018towards} ::
[Query: $\exists \Diamond ( \mathit{S1o.S1oEnd} ) $]
$\blacktriangleleft$} \\
\midrule
\rowcolor{white}
5  & 0.001 & 0.014 & 12.25 & 346
   & ${\boldsymbol \varepsilon}$ & \textbf{0.010} & \textbf{6.50} & \textbf{27}
   & 0.001 & 0.056 & 19.18 & 18 \\
\midrule
\multicolumn{13}{c}{$\blacktriangleright$
{\normalsize \simple{}}~\cite{jensen2023dynamic} ::
[Query: $\exists \Diamond ( \mathit{false} ) $]
$\blacktriangleleft$} \\
\midrule
\rowcolor{white}
1  & $\varepsilon$ & 0.013 & 7.64 & 56
   & ${\boldsymbol \varepsilon}$ & \textbf{0.006} & \textbf{5.81} & \textbf{3}
   & $\varepsilon$ & 0.053 & 18.76 & 3 \\
\rowcolor{lightgray}
1  & $\varepsilon$ & 0.013 & 7.92 & 428
   & ${\boldsymbol \varepsilon}$ & \textbf{0.006} & \textbf{5.90} & \textbf{3}
   & $\varepsilon$ & 0.054 & 18.76 & 3 \\
\rowcolor{white}
1  & 0.001 & 0.015 & 12.28 & 4028
   & \textbf{0.001} & \textbf{0.007} & \textbf{6.52} & \textbf{3}
   & 0.001 & 0.054 & 18.83 & 3 \\
\midrule
\multicolumn{13}{c}{$\blacktriangleright$
{\normalsize \soldiers{}}~\cite{farkas2018towards} ::
[Query: $\exists \Diamond ( E.\mathit{Escape} ) $]
$\blacktriangleleft$} \\
\midrule
\rowcolor{white}
5  & 0.128 & 0.146 & 494.43 & 24194
   & ${\boldsymbol \varepsilon}$ & \textbf{0.007} & \textbf{6.45} & \textbf{93}
   & $\varepsilon$ & 0.056 & 19.05 & 93 \\
\midrule
\multicolumn{13}{c}{$\blacktriangleright$
{\normalsize \srlatch{}}~\cite{farkas2018towards} ::
[Query: $\exists \Diamond ( \mathit{Env.envFinal} ) $]
$\blacktriangleleft$} \\
\midrule
\rowcolor{white}
3  & $\varepsilon$ & 0.013 & 7.62 & 9
   & ${\boldsymbol \varepsilon}$ & \textbf{0.007} & \textbf{6.24} & \textbf{5}
   & $\varepsilon$ & 0.052 & 19.05 & 4 \\
\bottomrule
\end{tabular}
}
\end{table}
}
}

\autoref{tab:punctual_full}, \autoref{tab:closed_ta_full}, \autoref{tab:big_constants_full}, and \autoref{tab:full_state_space_full} present the extended version of the benchmarks discussed in \autoref{sec:implementation_and_experimental_evaluation}.
In particular, \autoref{tab:full_state_space_full} includes an additional benchmark on the Lynch-Shavit protocol, similar to the Fischer case.

\textbf{Symmetry reduction}~\autoref{tab:symmetry} shows benchmarks in which the symmetry reduction optimization was enabled in \toolname{} and Uppaal (symmetry reduction cannot be manually controlled in TChecker); all symmetric \acp{ta} were identical.
Each $\mathit{Viking}_i$ had $20$ as its maximum constant.
The \train{} benchmark was executed using \ac{bfs}, as this exploration strategy yielded the best performance across all three tools.
Even though \toolname{} is not the fastest in \ac{vt}, symmetry reduction makes these benchmarks tractable with regions.  
Indeed, without this optimization, \toolname{} ran out of memory in \train{} for $K = 10, 14, 18, 22$, and in \bridge{} for $K = 17, 21$.  
Interestingly, the results suggest that TChecker does not apply specific optimizations in \train{}, but it does in \bridge{}.

\textbf{Additional benchmarks}~\autoref{tab:others} reports additional benchmarks that were conducted.  
These benchmarks include both small and medium-sized networks of \acp{ta}, as well as single \acp{ta}.  
The \exSITH{} and \simple{} benchmarks consist of a single \ac{ta}.
For \simple{}, the maximum constant was varied (10, 100, and 1000).
Apart from the \soldiers{} benchmark, \toolname{} performs comparably to both Uppaal and TChecker in terms of \ac{vt}, \ac{et}, and memory usage.
The query $\exists \Diamond ( \mathit{false} ) $ required to compute the entire reachable state space.
}
{\newpage
\appendix
\section{Artifact Evaluation}

We provide in this Appendix the details for the Artifact Evaluation Committee.
We also recall that the submitted paper is a \textbf{regular} paper.

\subsection{Badge claims}
We claim the following badges:
\begin{itemize}
    \item \emph{Artifact available}: our artifact has been uploaded to Zenodo~\cite{manini_2026_tarzan} and can be downloaded at: 
    \href{https://doi.org/10.5281/zenodo.18656202}{https://doi.org/10.5281/zenodo.18656202}
    \item \emph{Artifact functional}: the artifact allows to reproduce the experimental results presented in the paper. Additional details are given below.
\end{itemize}

\subsubsection{Functional outcomes}
The following functional outcomes are expected:

\textbf{F}$_1$: we first compared the performance of \toolname{} with Uppaal and TChecker in forward reachability.
The generated \texttt{benchmark\_results.pdf} file corresponds to \autoref{tab:punctual}, \autoref{tab:closed}, \autoref{tab:constants}, and \autoref{tab:safety}. 

\textbf{F}$_2$: we then compared forward and backward reachability in \toolname{}.
Generated files in the \texttt{backwards\_tarzan\_results} directory correspond to \autoref{tab:backward_scalability}.

\subsection{Quick start}
The artifact is a Docker image based on Ubuntu 24.04.
Both an ARM and an x86 version are available.
Be sure to run the right version based on your current architecture.
Once the right version is identified, it is possible to import it into your local Docker distribution using the following command:

\begin{center}
\begin{verbatim}
docker load -i <image_name>
\end{verbatim}
\end{center}
where \texttt{<image\_name>} is the name of the image you want to import.
Then, the ARM version can be run with the following command:
\begin{center}
\begin{verbatim}
docker run -it \
-v </path/to/your/shared/directory>:/artifact/output \
-w /artifact <image_name>
\end{verbatim}
\end{center}
while the x86 version with the following command:
\begin{center}
\begin{verbatim}
docker run --platform linux/amd64 -it \
-v </path/to/your/shared/directory>:/artifact/output \
-w /artifact <image_name>
\end{verbatim}
\end{center}
where \texttt{</path/to/your/shared/folder>} is a path to a directory on your machine where the output results will be saved, while \texttt{<image\_name>} is the name of the image you previously imported.
   
Once the Docker container is running, the initial tree directory structure should look like \autoref{fig:folder_structure}.
\begin{figure}[tb]
\dirtree{%
.1 artifact/.
.2 LICENSE.
.2 README.md.
.2 output.
.2 precomputed\_results.
.3 output.
.4 TARZAN.
.4 backwards\_tarzan\_results.
.4 benchmark\_summary.txt.
.4 benchmark\_tarzan\_results.
.4 benchmark\_tchecker\_results.
.4 benchmark\_uppaal\_results.
.4 latex\_results.
.3 running\_time.txt.
.2 scripts.
.3 run\_full.sh.
.3 run\_small.sh.
}
\caption{Artifact directory structure.}
\label{fig:folder_structure}
\end{figure}
In particular, the \texttt{output} directory will contain the output produced by the Artifact. 
Note that, once experiments are completed, this directory should have the same structure as the \texttt{output} directory located in \texttt{precomputed\_results}.
The \texttt{precomputed\_results} directory contains already executed experiments (to be used for comparison or in case some problems arise).
The \texttt{scripts} directory contains the scripts used to reproduce the experimental results.
For additional details on how to run the experiments, please refer to the \texttt{README.md} file.

The actual \toolname{} build can be found in the directory reachable via the following command:
\texttt{cd ../TARZAN}.
Everything is automated using the scripts contained in the \texttt{scripts} directory.
Hence, it is possible to ignore the \toolname{} build directory, but we encourage to take a look at it out of curiosity.

\subsection{Functional evaluation}

Both functional outcomes \textbf{F}$_1$ and \textbf{F}$_2$ can be fulfilled by using the \texttt{run\_small.sh} script.
This allows to reproduce \autoref{tab:punctual}, \autoref{tab:closed}, \autoref{tab:constants}, and \autoref{tab:safety} for \textbf{F}$_1$, and \autoref{tab:backward_scalability} for \textbf{F}$_2$.

The results of the artifact may vary from the ones reported in the paper depending on the physical resources (mostly memory) that are available.
Indeed, the different memory management between the Ubuntu Docker container and MacOS (on which the original experiments were conducted) may lead to out of memory events or timeouts that are not present in the paper's results.
For this reason, it is advisable to compare the reproduced results with the ones provided in the extended version of the paper, which contains also smaller instances of the experiments.
\todoAM{}{Citare extended version}}

\end{document}